\pdfoutput=1
\documentclass[10pt,twocolumn,letterpaper]{article}

\usepackage[letterpaper,textwidth=516pt,columnsep=12pt,top=0.75in,bottom=0.9in]{geometry}

\newif\ifsubmission
\submissionfalse

\usepackage{amsmath,amssymb,amsfonts}
\usepackage{amsthm}
\newtheorem{proposition}{Proposition}
\newtheorem{theorem}{Theorem}
\theoremstyle{definition}
\newtheorem{definition}{Definition}
\theoremstyle{plain}
\usepackage{graphicx}
\usepackage{booktabs}
\usepackage{algorithm}
\usepackage{algpseudocode}
\usepackage{xcolor}
\usepackage[hidelinks]{hyperref}
\usepackage[capitalise]{cleveref}
\usepackage{orcidlink}
\usepackage{tikz}
\usetikzlibrary{arrows.meta,positioning,fit,backgrounds,calc,shapes.geometric}
\usepackage{float}

\floatstyle{ruled}
\newfloat{functionality}{tbp}{lof}
\floatname{functionality}{Functionality}
\crefname{functionality}{Functionality}{Functionalities}
\Crefname{functionality}{Functionality}{Functionalities}
\crefname{section}{Section}{Sections}
\Crefname{section}{Section}{Sections}
\crefname{subsection}{Section}{Sections}
\Crefname{subsection}{Section}{Sections}
\crefname{subsubsection}{Section}{Sections}
\Crefname{subsubsection}{Section}{Sections}
\crefname{appendix}{Appendix}{Appendices}
\Crefname{appendix}{Appendix}{Appendices}

\definecolor{paperblue}{HTML}{5B7FA6}     %
\definecolor{paperbluelt}{HTML}{C9D6E3}   %
\definecolor{paperneutral}{HTML}{5B7B6E}  %
\definecolor{paperneutrallt}{HTML}{C7D0D6}%
\definecolor{paperamber}{HTML}{D4826A}    %
\definecolor{blueA}{HTML}{D3DEEA}
\definecolor{blueB}{HTML}{A9BDD3}
\definecolor{blueC}{HTML}{7B9BBF}
\definecolor{blueD}{HTML}{5B7FA6}
\definecolor{sanzotan}{HTML}{C6A87D}
\definecolor{sanzotanlt}{HTML}{D4C5A9}
\definecolor{sanzotangrey}{HTML}{B5A88B}
\definecolor{sanzobluegrey}{HTML}{8B9EA8}
\definecolor{sanzosage}{HTML}{7B9E87}

\newcommand{\Nc}{N}                                  %
\newcommand{\Dj}{\mathcal{D}_j}                      %
\newcommand{\nj}{n_j}                                %
\newcommand{\Cc}{C}                                  %
\newcommand{\featx}{\varphi}                         %
\newcommand{\thz}{\theta_0}                          %
\newcommand{\disp}{\Delta_j}                         %
\newcommand{\wj}{w_j}                                %
\newcommand{\thstar}{\theta^{\star}}                 %
\newcommand{\tc}{t}                                  %
\newcommand{\Fhe}{\mathcal{F}^{\mathsf{HE}\text{-}\mathsf{OFT}}}                       %
\newcommand{\Leak}{\mathcal{L}}                      %
\newcommand{\Dset}{\{\Dj\}_{j=1}^{\Nc}}              %
\newcommand{\negl}{\mathrm{negl}}                    %

\newcommand{\Asel}{A_{\mathrm{sel}}}                 %
\newcommand{\Aloc}{A_{\mathrm{loc}}}                 %
\newcommand{\Adis}{A_{\mathrm{dis}}}                 %
\newcommand{\Apool}{A_{\mathrm{pool}}}               %

\newif\iffromscratch\fromscratchtrue

\newcommand{\runin}[1]{\textbf{#1}\enspace}

  \newcommand{\paperonly}[1]{}

\makeatletter
\renewcommand\subsection{\@startsection{subsection}{2}{\z@}%
  {-3.25ex\@plus -1ex \@minus -.2ex}{1.5ex \@plus .2ex}%
  {\normalfont\large\bfseries\raggedright}}
\makeatother

\newenvironment{IEEEkeywords}{\par\smallskip\noindent\textbf{Keywords:}\ }{\par}

\hypersetup{pdftitle={HE-OFT: Privacy-Preserving One-Shot Federated Fine-Tuning
  under Homomorphic Encryption (Technical Report)},
  pdfauthor={Halil İbrahim Kanpak, Sinem Sav, Alptekin Küpçü}}
\begin{document}

\title{HE-OFT: Privacy-Preserving One-Shot Federated Fine-Tuning under Homomorphic Encryption\\[2pt]\large Technical Report}

\author{
Halil \.{I}brahim Kanpak\,\orcidlink{0009-0005-0284-5941},
Sinem Sav\,\orcidlink{0000-0001-9096-8768},
Alptekin K\"up\c{c}\"u\,\orcidlink{0000-0003-2099-2206}%
\thanks{Corresponding Author: Sinem Sav (sinem.sav@cs.bilkent.edu.tr)}%
\thanks{Halil \.{I}brahim Kanpak is with the Department of Computer Engineering, Ko\c{c} University, İstanbul, Türkiye (e-mail: hkanpak21@ku.edu.tr).}%
\thanks{Sinem Sav is with the Department of Computer Engineering, Bilkent University, Ankara, Türkiye (e-mail: sinem.sav@cs.bilkent.edu.tr).}%
\thanks{Alptekin K\"up\c{c}\"u is with the Department of Computer Engineering, Ko\c{c} University, İstanbul, Türkiye (e-mail: akupcu@ku.edu.tr).}%
}

\maketitle

\begin{abstract}
Many organizations adapt large pretrained models to their own tasks by fine-tuning on private data. Several of these parties often hold data for the same task and wish to fine-tune a model together without pooling that data.
Federated learning (FL) enables joint fine-tuning, but reconstruction attacks on shared intermediate values (the model or its gradients) remain a privacy risk. 
A one-shot protocol that exchanges one encrypted contribution exposes no intermediate value.
Such a protocol still gives the trained model to every participant, which is not permitted
where the model is a regulated or proprietary asset. We present HE-OFT\footnote{The implementation of HE-OFT, covering local
training, aggregation, selection and encrypted serving, is available at
\url{https://github.com/CRYPTO-KU/HE-OFT}.}, the first cryptographically secure one-shot federated
fine-tuning protocol in which no party receives the trained model. Each client
fine-tunes a low-rank adapter and a classifier head on a frozen public backbone
and keeps the adapter. The client uploads one encrypted head displacement, which the
server combines under multiparty CKKS and never decrypts. A quorum of clients
returns only the predicted label to the querier. On four text
classification tasks and one vision task, HE-OFT reaches $61$ to $79$ per cent accuracy, against $20$ to $48$ per cent for a client training alone. HE-OFT keeps $85$ to $96$ per cent of the accuracy of a disclosed model. A test-time query takes $443.1$ to $1713.1$\,s on one core, or $56.1$ to $255.1$\,s with level restoration on a GPU. Restoring levels at the server cuts the traffic per query from up to $1.6$\,GiB to $13.5$\,MiB.
\end{abstract}


\section{Introduction}
\label{sec:intro}

Organizations fine-tune a large public pretrained model, known as the
\emph{backbone}, on their own data~\cite{bommasani2021opportunities,hu2022lora}. Several
organizations often hold complementary data for the same task, such as different
patient populations, fraud patterns, or customer bases. Each would obtain a
better model from the combined data than from its own~\cite{sheller2020federated}.
Pooling the data is difficult, because data protection law restricts the sharing of
confidential records across institutions~\cite{gdpr2016,hipaa,kvkk2016,pipl2021,dpdp2023}. Competing organizations
also resist pooling, and handing the data to the operator of the fine-tuning
infrastructure meets the same objection~\cite{xiao2023offsite}.

\begin{figure*}[t]
\centering
\includegraphics{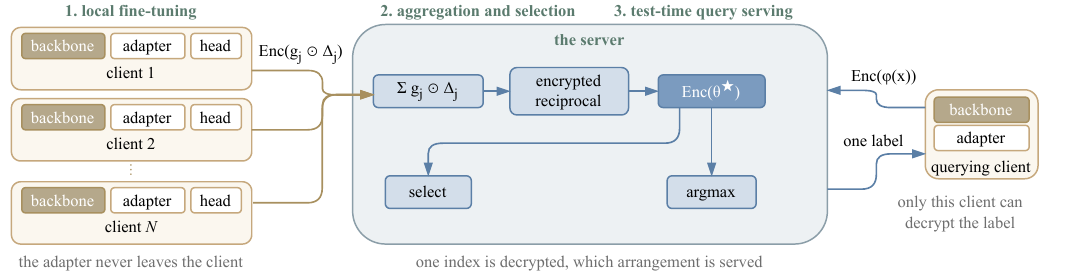}
\caption{HE-OFT in three stages: local fine-tuning, aggregation and selection, and test-time query serving. Each client fine-tunes an adapter and a
classifier head on the same frozen public backbone. The client keeps the adapter and
uploads one encrypted head displacement weighted by its per-class counts. The
server adds the ciphertexts and divides by the per-class totals under
encryption. The shared head $\thstar$ exists only as a ciphertext.}
\label{fig:training}
\end{figure*}

\runin{Privacy risks in federated learning.}Federated learning (FL) enables collaboration~\cite{mcmahan2017communication},
and its privacy risk is concentrated at training time~\cite{nasr2019comprehensive,melis2019exploiting}. FL exchanges
model updates in place of data, and the updates are themselves open to privacy
attacks. An adversary who aggregates a client's per-round updates performs
membership inference at a $25.3\%$ true-positive rate (TPR) at a $0.1\%$
false-positive rate (FPR), against $0.18\%$ from a single snapshot~\cite{zhu2025fedmia}. The attack surface thus depends on what a protocol reveals in training and how often.

\runin{Requirements for privacy-preserving federated fine-tuning.}A federated fine-tuning protocol that resists these attacks needs three properties. The first is a \emph{single exchange}. Each further round gives an adversary another observation of the same
client, and repetition alone raises the true-positive rate of membership inference by two orders of magnitude. Such a protocol is known as \emph{one-shot}~\cite{guha2019oneshot}. Each client
contributes once, and no intermediate training artifact is exposed. The parties communicate afterwards to answer queries. The second is \emph{encrypted
contributions}, since even a single update can expose its training
data~\cite{zhu2019deep,geiping2020inverting}. The third is a \emph{confidential
model}, which no party holds in plaintext. A model fine-tuned on confidential
records can be a regulated artifact in domains such as health, biometrics, and
credit~\cite{euaiact2024}. If the released model averages the
contributions with public weights, all clients but one can also subtract their
own and recover the last.

\runin{Threat scenario.}The parties are the clients and one server that combines
their contributions and answers queries. The adversary corrupts the server
together with any coalition of fewer clients than the threshold $\tc$ required
for decrypting a ciphertext. The corrupted parties
follow the protocol, pool what they observe, and try to learn an honest client's
data, its class histogram, or the trained model. We also let the corrupted clients
deviate from the protocol while the server stays honest. A deviating server must
be excluded, because it could present any ciphertext for decryption, such as a
row of the shared head, and the clients cannot tell it from a legitimate request
(\cref{prop:nogate}).

\runin{Existing approaches.}Each of the closest systems meets some of these properties
(\cref{tab:related}). One-shot FL~\cite{guha2019oneshot} methods such as
DENSE~\cite{zhang2022dense}, FuseFL~\cite{tang2024fusefl} and
GH-OFL~\cite{turazza2026ghofl} exchange once, send the contribution in plaintext
and give every participant the model. FedAUXfdp adds differential privacy (DP) to
the one-shot exchange and trades accuracy for a tighter privacy
budget~\cite{hoech2022fedauxfdp}. POSEIDON keeps the contributions and the model
under homomorphic encryption (HE) over many rounds, with every activation replaced
by a polynomial in a deep, repeatedly bootstrapped circuit~\cite{sav2021poseidon}.
SHE-LoRA encrypts a selected part of each update over many rounds and shares each
round's aggregate with the clients~\cite{li2025shelora}. slytHErin and CryptPEFT
evaluate a model trained elsewhere, under HE or two-party computation (2PC), and
return scores to the client that sends the
query~\cite{intoci2023slytherin,xia2026cryptpeft}. HE-OFT combines a single
encrypted exchange with a model that no party holds in plaintext, and it returns
only the label.

\runin{Design principles.}The three properties together determine each stage of the design. Training under
encryption is impractical at this scale, and each client therefore trains alone,
in plaintext, and contributes once. One round suffices because every client
starts from the same pretrained backbone. Pretraining places the clients'
fine-tuned solutions in a close region of the loss landscape, and the changes
they make thus combine
linearly~\cite{ilharco2023editing,malinovsky2024oneround}. Heterogeneous data
still costs a single round more accuracy than many rounds, which correct the
merge step by step~\cite{tao2024taskarith}. The client that sends a query, called the \emph{querier},
computes the features of that query itself. Hence we place the shared trained map after the last nonlinearity. In a transformer, that map is the final linear map.
An earlier position would force the querier to evaluate nonlinearities under
encryption or to decrypt an intermediate value. The merge divides each class
row by the federation's total count for that class. The server keeps these totals
encrypted, because all clients but one could subtract their own counts from them
and read the last client's class histogram. A querier that receives scores for
inputs of its choice can recover a linear head from $d+1$ answers, where $d$ is
the feature dimension.
The federation therefore returns only the label, and only to the querier. 

\runin{Our approach.}We present HE-OFT, the first one-shot federated fine-tuning protocol in which no party
ever holds a contribution, a training artifact, or the trained model in
plaintext. Each client fine-tunes a small adapter and a classifier head on the
frozen public backbone. It encrypts its contribution
under multiparty CKKS~\cite{cheon2017ckks,mouchet2021multiparty}, a homomorphic
encryption scheme in which the secret key is split across the parties. Because
the backbone stays public and in plaintext, the cryptographic cost depends on the
feature dimension and the number of classes, and not on the depth or size of the
backbone. HE-OFT can thus fine-tune pretrained transformers.

\runin{Overview.}HE-OFT has three stages (\cref{fig:training}). In local fine-tuning, every client fine-tunes an adapter and a classifier head on the same frozen public backbone and
keeps the adapter. We call the difference between a client's trained head and a public
initializer its \emph{head displacement}. The client uploads one encrypted head
displacement, weighted by its per-class counts. In aggregation and selection, the server adds the ciphertexts and multiplies the
sum by the encrypted reciprocal of the per-class totals. The result is the shared
head, in which each class row comes from the clients that hold that class. The
shared head exists only as a ciphertext, and the server holds no key to decrypt
it. The federation then chooses between the shared head over the bare backbone
and the shared head over each client's own adapter, called the two
\emph{arrangements}. The federation decrypts only the
index of the arrangement it serves. Both arrangements cost the same per query,
because the querier runs the backbone and its own adapter in plaintext. In test-time query serving, the querier encrypts its features. The server applies the shared head and takes the argmax under encryption. A quorum of clients re-encrypts the label under the querier's key, and only the querier can decrypt it.

\runin{Contributions.}We make the following contributions.
\begin{itemize}
\item HE-OFT, a one-shot federated fine-tuning protocol under multiparty CKKS in
  which the trained model is never decrypted (\cref{sec:method}).
\item A merge that weights each class row of the encrypted head displacements
  by the clients' counts for that class. The server inverts the class totals
  under encryption, and every adapter stays with its owner (\cref{sec:split}).
\item An encrypted serving path for fine-tuned transformers under multiparty HE. The
  server computes the argmax without decrypting, and a quorum of clients returns
  only the label, and only to the querier (\cref{sec:serving}).
\item An estimator that selects between the two arrangements under encryption and
  decrypts one index (\cref{sec:selection}).
\item A simulation proof against semi-honest parties, an input-privacy bound
  against malicious clients, and a proof that an honest server is necessary
  whenever the clients decrypt ciphertexts that the server presents
  (\cref{sec:security}).
\item An implementation in multiparty CKKS, evaluated with pretrained
  RoBERTa-base~\cite{liu2019roberta} and vision transformer
  ViT-B/16~\cite{dosovitskiy2021vit} backbones on five tasks (\cref{sec:experiments}). HE-OFT reaches $61$ to $79$ per cent accuracy where a client alone
  reaches $20$ to $48$ per cent. It keeps $85$ to $96$ per cent of the accuracy of
  a model aggregated and released in plaintext. A query costs $13.5$\,MiB of
  traffic.
\end{itemize}

\section{Preliminaries}
\label{sec:prelim}

\subsection{Collaborative Fine-Tuning}
\label{sec:prelim-ft}

Fine-tuning adapts a pretrained model to a new task with a modest amount of
labeled data~\cite{bommasani2021opportunities,hu2022lora}.
The frozen pretrained part is known as the \emph{backbone}, and we write $\featx$ for
the map from an input to its $d$-dimensional feature vector. Parameter-efficient
fine-tuning trains a small number of added parameters and leaves the backbone
unchanged. LoRA adds a low-rank update to a backbone weight matrix. The update is the
product of a down-projection to rank $r$ and an up-projection back to the full
width~\cite{hu2022lora}. We freeze the down-projection at a public initialization
and train the up-projection alone, as FFA-LoRA does~\cite{sun2024ffalora}. We
call this trained unit client $j$'s \emph{adapter} $A_j$, and $\varphi_j$ denotes
the backbone composed with it.

A classifier head is a linear map from the $d$ features to $\Cc$ class logits.
Every client starts its head from one public initializer $\thz$ and trains it
with its adapter. The \emph{head displacement} of client $j$ is
$\disp=h_j-\thz$, where $h_j$ is its trained head.

In federated learning (FL), $\Nc$ clients train a model on their joint data
without moving it~\cite{mcmahan2017communication}. Client $j\in\{1,\dots,\Nc\}$
holds a dataset $\Dj$ of size $\nj$, with $n_{j,c}$ examples of class
$c\in\{1,\dots,\Cc\}$. FL usually runs many rounds, and in each round the clients
send updates that the server averages. One-shot FL replaces the rounds with a
single exchange~\cite{guha2019oneshot}. In HE-OFT that exchange carries each
client's encrypted head displacement, from which the server forms the shared
head $\thstar$. We write $\mathsf{Enc}(\cdot)$ for encryption under the collective
public key of \cref{sec:mhe}. Notation appears in \cref{tab:notation}.

\subsection{Multiparty Homomorphic Encryption}
\label{sec:mhe}

Multiparty homomorphic encryption lets several parties compute on data that
stays encrypted under one key, which no party holds
alone~\cite{mouchet2021multiparty}. Each party keeps a share of the secret key,
and decryption needs the cooperation of enough parties. We build on the CKKS
instantiation~\cite{cheon2017ckks}, which encrypts vectors of real numbers and
computes on them approximately. CKKS is a leveled
scheme~\cite{lyubashevsky2010ideal}, and each multiplication consumes one
of a finite number of levels, which bootstrapping
restores~\cite{cheon2018bootstrapping}.

We use the $\tc$-out-of-$\Nc$ threshold access structure of Mouchet et
al.~\cite{mouchet2023threshold}, where $2\le\tc\le\Nc$. Any $\tc$ clients form a
\emph{quorum} that can decrypt, and no $\tc-1$ clients can. The corruption bound
throughout \cref{sec:security} is hence $\tc-1$.

HE-OFT uses four operations of the multiparty scheme. Key generation combines the
clients' key shares into one collective public key and the evaluation keys, and
no party ever forms the secret key itself. Any party can thus encrypt and compute
on ciphertexts, but decryption needs shares from a quorum, and fewer shares reveal
nothing about the plaintext. A key switch turns a ciphertext under the collective
key into one under a chosen party's public key, and only that party can then
decrypt it. A collective refresh restores the levels of a ciphertext, with every
client contributing a share.

Our implementation sets $\tc=\Nc$. We measure that setting in \cref{sec:exp-cost}. Every statement in \cref{sec:security} holds for any $\tc$.

\section{Method}
\label{sec:method}
\label{sec:setting}

The $\Nc$ clients of \cref{sec:prelim-ft} hold private labeled data over a
common label space of $\Cc$ classes. The class proportions differ across
clients, and a client may hold no example of a class. The clients seek one
classifier over the whole label space without pooling their data.

We consider $\Nc$ clients and one server. The clients are organizations that each
hold confidential labeled data, and the server is an untrusted service that
combines their contributions and answers test-time queries. A static, semi-honest
adversary corrupts the server together with up to $\tc-1$ clients. The corrupted
parties follow the protocol, pool their views, and try to learn an honest client's
data, its class histogram, or the shared head. We also consider up to $\tc-1$
clients that deviate arbitrarily while the server stays honest
(\cref{sec:malicious}). Every client may query the served head and learn the
label of each query, up to an allowance of $Q$ queries. In \cref{sec:security}, we
state what each party sees and prove the guarantees.

\begin{table}[t]
\centering
\caption{Notations.}
\label{tab:notation}
\footnotesize
\setlength{\tabcolsep}{4pt}
\begin{tabular}{@{}ll@{}}
\toprule
Symbol & Meaning \\
\midrule
$\Nc$, $\Cc$, $d$ & numbers of clients and classes, feature dimension \\
$\featx$, $\varphi_j$ & public backbone, and backbone with client $j$'s adapter $A_j$ \\
$\thz$, $h_j$, $\disp$ & public head initializer, client $j$'s head, $\disp=h_j-\thz$ \\
$g_j$ & client $j$'s per-class counts, $g_{j,c}=n_{j,c}$ \\
$\thstar$ & shared head, held only as a ciphertext \\
$\mathsf{Enc}$ & encryption under the collective public key \\
$\tc$ & decryption threshold \\
$Q$ & per-client query allowance \\
$\Dj,\ \nj$ & client $j$'s dataset and its size $\nj=|\Dj|$ \\
$n_{j,c}$ & examples of class $c$ held by client $j$ \\
$g_j\odot\disp$ & $\disp$ with each class row scaled by its count \\
$\wj$ & sample weight $\nj/\sum_i n_i$ \\
$\alpha$ & concentration of the Dirichlet label partition \\
$K$ & number of local steps \\
$\Aloc$ & accuracy of a client alone \\
$\Asel$ & accuracy of the selected arrangement \\
$\Adis$ & accuracy of the disclosed model \\
$\Apool$ & accuracy of pooled training \\
\bottomrule
\end{tabular}
\end{table}

Two limits of computation under encryption shape the protocol.
Gradient steps, softmax, and adaptive curvature need deep circuits that
bootstrap repeatedly, and all learning therefore happens in plaintext on the
clients. Multiplicative depth dominates the homomorphic cost, and the server
therefore combines the uploads by a linear map. A server that computes on
ciphertexts also cannot inspect, compare, or branch on a value. It cannot apply an
aggregation rule that depends on the data, and the clients thus choose between
the two arrangements under encryption.

\subsection{Encrypted One-Shot Training}
\label{sec:split}

\subsubsection{Local Fine-Tuning of the Adapter and Head}
Each client trains its adapter $A_j$ and its head $h_j$ on the frozen public
backbone $\featx$ (\cref{sec:prelim-ft}). Only the head is shared. Client $j$
computes its head displacement $\disp\in\mathbb{R}^{\Cc\times d}$, encrypts it,
and uploads it once. The server combines the encrypted
displacements into a single shared head and never decrypts the result. The
adapter stays on the client.

\subsubsection{Placement of the Shared Head}
The querier computes the features of its query itself. If the shared trained quantity were an adapter placed inside the
backbone, the features would depend on it, and computing them would require
evaluating the backbone with encrypted weights. The client would then have to
evaluate every nonlinearity of the network homomorphically. We therefore place the shared trained quantity after the last nonlinearity, as the final linear map of a transformer.

Any trained linear map after the last nonlinearity can take this place,
provided the querier can compute the features that enter it. A classifier head is one such map. The vocabulary projection of an autoregressive decoder is another. It is linear and follows the final layer normalization, and the client computes its input from a prefix it generated itself. We evaluate classification only. HE-OFT does not require the label space
to be a set of classes. \Cref{sec:scope} lists what a generation setting would
also need, including the greedy decoding that the argmax of \cref{sec:serving} imposes.

Each client keeps its own adapter, and the representation still adapts to local
data. The split between adapter and head follows from coverage. A client can
improve its own representation alone, but it cannot learn to recognize a class it
has never observed. Coverage is a property of the head, which the federation provides. \Cref{sec:exp-split} measures both parts. Keeping the
adapter local also matters cryptographically, because no aggregate of the
adapters is ever formed (\cref{sec:threat}).

\subsubsection{Merging Client Contributions by Class Coverage}
Let $g_{j}\in\mathbb{R}^{\Cc}$ collect client $j$'s per-class example counts,
$g_{j,c}=n_{j,c}$. The shared head is the coverage-weighted combination
\begin{equation}
\label{eq:aggregate}
\thstar \;=\; \thz \;+\; \frac{\sum_{j=1}^{\Nc} \mathrm{diag}(g_j)\,\disp}
                              {\sum_{j=1}^{\Nc} g_j},
\end{equation}
in which row $c$ of the shared head is the count-weighted average of the clients'
row-$c$ displacements. Clients holding a class decide its row, and rows that no
client covers remain at the public initializer. HE-OFT handles label skew through these weights. Row $c$ depends only on the clients with $n_{j,c}>0$, each in
proportion to $n_{j,c}$. Each client forms the product
$\mathrm{diag}(g_j)\,\disp$ locally before encryption, and the server therefore
only adds ciphertexts.

\subsubsection{Per-Class Averaging Under Encryption}
The denominator of \cref{eq:aggregate} is the vector of federation-wide per-class
totals, and it must stay encrypted. Every client knows its own counts. A
coalition of $\Nc-1$ clients that learned the totals could subtract its own counts
and read the remaining client's class histogram. The server therefore evaluates
the reciprocal under encryption, with an inversion circuit over the positive
domain whose refreshes are collective. The server inverts the totals once, at aggregation.

\begin{algorithm}[t]
\caption{Building the shared head. All clients hold the same frozen public
backbone $\featx$ and the same public head initializer $\thz$.}
\label{alg:train}
\small
\begin{algorithmic}[1]
\State clients run distributed key generation, obtaining a collective public key
\For{client $j=1,\dots,\Nc$ \textbf{in parallel}}
  \State train adapter $A_j$ and head $h_j$ on $\Dj$ \Comment{plaintext, local}
  \State \textbf{keep} $A_j$, which is never transmitted
  \State send $\mathsf{Enc}(\mathrm{diag}(g_j)\,\disp)$ and $\mathsf{Enc}(g_j)$,
         where $\disp=h_j-\thz$
\EndFor
\State \textbf{server} adds the ciphertexts to form the numerator and the
       denominator of \cref{eq:aggregate} \Comment{additions only}
\State \textbf{server} evaluates the reciprocal of the denominator under
       encryption and multiplies \Comment{once, not per query}
\State \Return $\mathsf{Enc}(\thstar)$, which no party decrypts
\end{algorithmic}
\end{algorithm}

The inversion
converges on a bounded positive interval. We scale the totals by the
federation's example count, which is public, and the argument thus lies in
$(2^{-6},1]$. A class whose share is below $2^{-6}$ falls outside that interval.
We measured a relative error of $1.9\times10^{-8}$ at ten parties, on denominators in $[0.02,0.99]$. This range does not cover the lowest shares.
Banking77 and CIFAR-100 have class shares near $0.013$ and $0.010$. A deployment
at those label-space sizes can scale the totals by the example count divided by $\Cc$, also a public parameter. A balanced class then maps to $1$ instead of $1/\Cc$. A deployment can instead widen the interval to $(2^{-k},1]$ with $k$ set from
$\Cc$. Either option adds depth at aggregation and costs nothing per query.

A class that no client holds would cause a division by zero.
Detecting a zero total under encryption needs a comparison, which HE-OFT does not include.
In \cref{sec:exp-split}, we count the holders of every class row and find none at zero on any task, seed or partition we measured. The case therefore does not arise in our configuration. A deployment that cannot rule it out needs either a public list of covered classes or a comparison in the aggregation circuit. The list discloses coverage, and the cost of \cref{sec:exp-cost} does not include the comparison.

\subsubsection{Cost of the Encrypted Aggregation}
The server performs one ciphertext addition per client per ciphertext, one
encrypted reciprocal, and one multiplication. The encrypted object is the head
alone, whose size is set by the number of classes and the feature dimension.
Communication is one upload per client, and no client downloads model
parameters. \Cref{alg:train} gives the full training procedure.

\subsection{Serving Arrangement Selection Under Encryption}
\label{sec:selection}

\begin{figure*}[t]
\centering
\includegraphics[width=\textwidth]{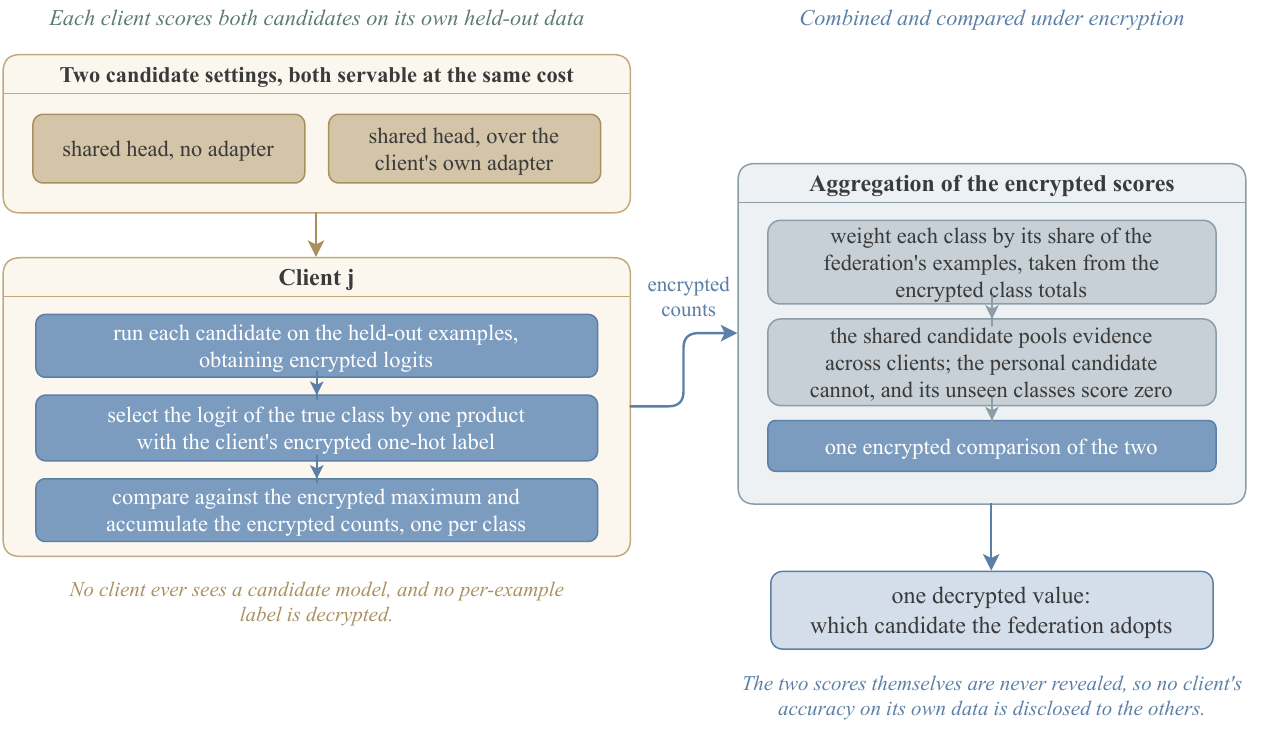}
\caption{Choosing between two arrangements without decrypting either. Each client scores both on data held back from its own training. One product with the client's encrypted one-hot label selects the true-class logit. The per-class counts of correct predictions are weighted by each class's share of the federation's examples, taken from the count matrix already used for the aggregation. The shared arrangement pools evidence across clients because it is one model, and the personal arrangement does not. Only the classes a client holds add to its score, which is thus a lower bound. A quorum decrypts one value only.}
\label{fig:selection}
\end{figure*}

\subsubsection{Encrypted Scoring of the Shared Head With and Without the Client Adapter}
Two arrangements are available at the same cost. We call the shared head over
each client's own adapter the personal arrangement, and the shared head over the
bare public backbone the shared arrangement. Which is better depends on the task, and the federation chooses between
them without decrypting either. Each client evaluates both, under encryption, on
data it held back from its own training, and reports encrypted per-class counts.
The server combines the counts into one encrypted score per arrangement. The
federation decrypts one value, the index of the arrangement it adopts. \Cref{fig:selection} shows the steps.

\subsubsection{Prior-Weighted Estimator}
A simpler rule lets each client measure both arrangements on its held-out data
and vote. That rule estimates the wrong quantity. It selects the personal
arrangement in every cell and is thus wrong wherever
the shared arrangement is better (\cref{sec:exp-select}).
\Cref{alg:select} gives the full procedure. We score both arrangements by
the expected accuracy of a prediction on an example drawn from the federation's
own class distribution,
\begin{equation}
\label{eq:estimator}
E \;=\; \sum_{c=1}^{\Cc} p_c\,a_c,
\qquad
p_c \;=\; \frac{\sum_j n_{j,c}}{\sum_j \nj},
\end{equation}
where $a_c$ is the accuracy on class $c$ and $p_c$ is that class's share of the
federation's examples, taken from the per-class counts already used in
\cref{eq:aggregate}. The shared arrangement is one model, and per-class
evidence from different clients therefore concerns the same object and can be
combined. In the personal arrangement, each client's evidence concerns only its
own model, and classes it does not hold contribute nothing. Scoring those classes
at zero makes $E$ a lower bound for the personal arrangement. The estimator has no fitted threshold. It chooses well in every cell we measured (\cref{sec:exp-select}). A simpler rule is sometimes right where the estimator is not, but none is as consistent. The estimator also runs entirely under encryption.

\begin{algorithm}[t]
\caption{Choosing an arrangement without decrypting either. The federation learns only the index of the winner.}
\label{alg:select}
\small
\begin{algorithmic}[1]
\For{each arrangement $m$ and each client $j$}
  \For{each held-out example $(x,y)$ of client $j$}
    \State obtain $\mathsf{Enc}(\ell)$ and $\mathsf{Enc}(\max_c \ell_c)$ as in
           \cref{alg:serve}
    \State select $\mathsf{Enc}(\ell_{y})$ by one product with $\mathsf{Enc}(e_y)$
           \Comment{client $j$ sends its label as an encrypted one-hot vector $e_y$}
    \State add $\mathsf{Enc}(e_y)\cdot\mathsf{Enc}\big(\mathbf{1}[\ell_y = \max_c \ell_c]\big)$
           to the per-class accumulator \Comment{$y$ stays encrypted}
  \EndFor
  \State pad the upload of client $j$ to a public size with zero labels
\EndFor
\State combine the accumulators into $\mathsf{Enc}(E_m)$ by \cref{eq:estimator},
       using public scalars \Comment{one level}
\State compare the two encrypted scores
\State zero every slot of the comparison but the first, with a plaintext mask
\State \Return the decrypted index of the larger \Comment{one value}
\end{algorithmic}
\end{algorithm}

\subsubsection{Accuracy Estimation per Class}

The estimator needs an estimate of accuracy on classes a client has not
seen, and its only local proxy is accuracy on classes it has barely seen. A
held-out set carved uniformly almost never contains a class of which the client
holds one or two examples, and that proxy is then unmeasurable. Each client
therefore reserves at least one held-out example from every class it holds before
carving the remainder at random. A client with a single example of a class puts
that example in the held-out set. Its adapter then never trains on that class,
and the estimator thus obtains the measurement it needs. The holdout costs little training data and leaves the cryptographic protocol unchanged.

\subsubsection{Selection Cost and Resulting Disclosure}
Apart from the comparisons, every step of the selection needs at most one
homomorphic multiplication, and therefore at most one level. The comparisons use
the same circuit as the argmax of \cref{sec:serving}. One value is decrypted for the whole federation. The
per-client per-class accuracies are not decrypted. The selection runs within the
threat model of \cref{sec:method}. Every step is a computation on ciphertexts or
the decryption of the index $a^{\star}$, which $\Fhe$ announces. Against deviating
clients, $a^{\star}$ is one more output that depends on the honest data, and
\cref{def:contentgame} therefore compares only datasets that lead to the same index.

\subsection{Test-Time Query Serving}
\label{sec:serving}

\begin{figure*}[t]
\centering
\includegraphics[width=\textwidth]{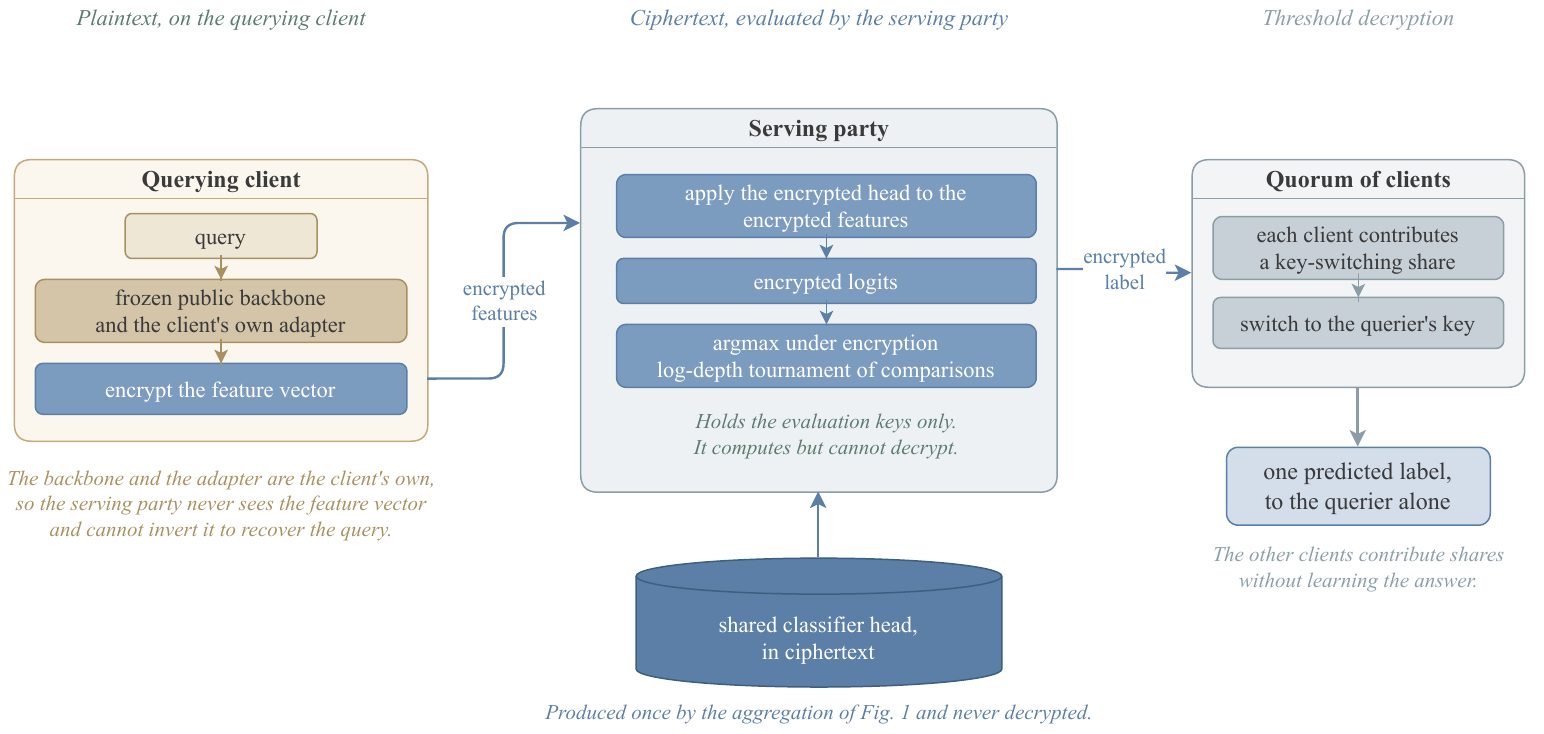}
\caption{Serving one test-time query. The querier computes the features with its own backbone and adapter and encrypts them. The server thus never sees the feature vector and cannot invert it to recover the query. The server holds the evaluation keys only. It applies the encrypted head, obtains encrypted logits, and reduces them to the index of the largest entry by a tournament of homomorphic comparisons. A quorum of clients then re-encrypts that index under the querier's key. Quorum members other than the querier hence contribute shares without learning the index.}
\label{fig:serving}
\end{figure*}

\Cref{fig:serving} shows the serving path. HE-OFT serves a trained quantity that is linear in the features. Secure transformer inference addresses the problem of serving a trained quantity placed inside the backbone. An encrypted shared adapter composes
with those systems at the cost they report, and we place that case outside our
scope. Those systems evaluate a plaintext model on encrypted inputs. Here the shared
weights are encrypted as well, and each adapter site would thus need a
ciphertext-by-ciphertext product. The nonlinearities dominate the cost in that
literature and stay unchanged, and hence the added cost is one low-rank product per
site.

\subsubsection{Test-Time Query Protocol}
The querier computes the features of its query with its own
backbone and adapter, encrypts them, and sends the ciphertext to the server. The
server applies the encrypted shared head to the encrypted features, obtains
encrypted logits, and reduces them to the index of the largest entry under
encryption. \Cref{alg:serve} gives the steps.

\subsubsection{Test-Time Query Encryption}

Sending the features in plaintext would make the head application
cheaper. We encrypt them because features can be inverted. A party holding
$\varphi_j(x)$ and the public backbone can reconstruct an approximation of $x$,
and plaintext features would therefore give the server the client's input.
Encrypted features make the server apply the head ciphertext by ciphertext, and
\cref{sec:exp-cost} reports the cost.

\subsubsection{Label-Only Output}

Returning logits would expose the head. The head is a linear map on
features the client computes itself. A client that reads logits for feature
vectors of its choosing therefore recovers the head by solving a linear system,
in as many queries as the feature dimension~\cite{tramer2016stealing}. The same
object, the final projection layer of a production transformer, has been
recovered through a deployed interface for a few tens of dollars. That interface
exposed log-probabilities and a logit bias, and its operators then restricted
that combination~\cite{carlini2024stealing}. Returning only the index of the
largest entry imposes the same restriction in the protocol itself. This restriction places an adversary in the hard-label
setting, where functionally equivalent extraction has only recently been achieved
for general networks~\cite{chen2024hardlabel,carlini2025hardlabel}. Both attacks
target networks with hidden layers and spend their effort on the hidden neurons. Our shared quantity
has no hidden neurons, and the restriction thus helps less here than those results suggest. The restriction replaces the linear solve with a search for decision boundaries, at a
cost logarithmic in the precision demanded. It also leaves the head determined
only up to a positive scale and a vector added to every row. Extraction remains
possible, and \cref{sec:threat} states the guarantee as a bound on queries and
not as a hardness claim.

\subsubsection{Multiparty Decryption Among the Clients at Test Time}
Decryption needs a quorum, and clients other than the querier take part in
every answer. They contribute key-switching shares and read nothing, because the
quorum switches the ciphertext to the querier's key without decrypting the
label.

\begin{algorithm}[t]
\caption{Answering one query. The shared head remains a ciphertext throughout.}
\label{alg:serve}
\small
\begin{algorithmic}[1]
\Require query $x$ at client $j$ and $\mathsf{Enc}(\thstar)$ held by the server
\State client $j$ computes $\varphi_j(x)$ with its own backbone and adapter
       \Comment{plaintext, local}
\State client $j$ sends $\mathsf{Enc}(\varphi_j(x))$
\State \textbf{server} computes
       $\mathsf{Enc}(\ell) \gets \mathsf{Enc}(\thstar)\,\mathsf{Enc}(\varphi_j(x))$
       \Comment{ciphertext by ciphertext}
\State \textbf{server} reduces $\mathsf{Enc}(\ell)$ to
       $\mathsf{Enc}(\hat{y})$, $\hat y=\arg\max_c \ell_c$, by a tournament of
       homomorphic comparisons~\cite{cheon2020comparison}
\State \textbf{server} zeroes every slot of $\mathsf{Enc}(\hat y)$ but the first,
       with a plaintext mask \Comment{only the label remains}
\State a quorum of clients key-switches $\mathsf{Enc}(\hat y)$ to client $j$'s
       public key
\State \Return $\hat y$, readable by client $j$ alone
\end{algorithmic}
\end{algorithm}

\subsubsection{Test-Time Query Cost}

The argmax is the only operation in HE-OFT that costs more than one
level, because it compares logits. Evaluated
as a sequential fold of $\Cc-1$ comparisons, it costs minutes for a large label
space. Packing the logits into one ciphertext and reducing them by a tournament
lowers the number of sequential comparisons from $\Cc-1$ to
$\lceil\log_2\Cc\rceil$, and \cref{sec:exp-cost} reports the measured result.
Homomorphic evaluation
under a collectively generated key works as under a single key, and the two
differ only in key generation and decryption. Single-key results for the argmax therefore carry
over. The server holds the collective public key, the evaluation keys, and
ciphertexts. It can thus compute but not decrypt, and it learns neither the model
nor any client's data. The server is trusted to follow the protocol (\cref{prop:nogate}) and to stay
available, and verifiable
homomorphic encryption applies where honest computation must be
certified~\cite{viand2023verifiable}.

\section{Security}
\label{sec:security}

\label{sec:threat}

The participants are $\Nc$ clients and a server. A static, semi-honest adversary
corrupts the server together with up to $\tc-1$ clients. The corrupted parties
follow the protocol, pool their views, and try to infer private information from
what they observe. In \cref{sec:malicious}, we let up to $\tc-1$ clients deviate arbitrarily while the server stays honest. We prove in \cref{sec:malicious-ext} that this assumption is necessary whenever the
clients decrypt ciphertexts that the server presents.

The server sees ciphertexts under the collective public key, the
public per-client sample counts, and the evaluation keys. It does not hold the decryption
key. The server also receives
the encrypted query features.

Three properties of HE-OFT restrict what a coalition can compute.
Every adapter stays with its owner, and HE-OFT forms no aggregate from which a
coalition could subtract its own contributions. The shared head stays a
ciphertext throughout, and a coalition thus cannot perform the same subtraction
on it. The server inverts the per-class totals under encryption and never
decrypts them. HE-OFT hence discloses no plaintext value from which a coalition could
read the remaining client's class histogram.

The guarantee covers disclosure only, and we do not prove that the
histogram cannot be recovered. A client's counts also shape the displacement it
uploads. For a linear map trained under cross-entropy, the row of a class grows
in magnitude with the number of examples of that class the client
held~\cite{wainakh2022user}. Those
displacements and the merged head are ciphertexts, whose plaintexts IND-CPA
security hides. \Cref{sec:c2c} measures what an extracted copy of the head
reveals. Inverting the totals under encryption removes
the only channel that would have disclosed the histogram in plaintext.

\subsection{The Ideal Functionality}
\label{sec:ideal}

\Cref{func:ideal} defines the ideal functionality $\Fhe$ against which we
measure HE-OFT. $\Fhe$ takes its form from the secure aggregation functionality of
ELSA~\cite{rathee2023elsa}. As in FULLSA~\cite{karakoc2024fullsa}, the server is a
party outside $\Fhe$. It receives only the selected index and the phase signals,
and the semi-honest adversary of \cref{sec:threat} may corrupt it.

\begin{functionality}[t]
\caption{Ideal functionality $\Fhe$}
\label{func:ideal}
\footnotesize
\noindent\textbf{\emph{Parameters.}} The client count $\Nc$, the decryption
threshold $\tc$, the per-client query allowance $Q$, and the label space size
$\Cc$. The parties are clients $P_1,\dots,P_{\Nc}$ and a server $S$, which is not
part of $\Fhe$. The functionality holds a
counter $q_j$ for each client, each initialized to $0$, and a store that is
initially empty and that $\Fhe$ never sends to any party.

\noindent\textbf{\emph{Inputs.}} Each client $P_j$ inputs its dataset $\Dj$, and
later inputs queries $x$. The server $S$ inputs a selection request.

\noindent\textbf{\emph{Outputs.}} Every party receives the selected index
$a^{\star}$ and the phase signals. Client $P_j$ receives one label for each of
its first $Q$ queries, and $\bot$ afterwards.

\noindent\textbf{\emph{Steps.}}
\begin{enumerate}\itemsep2pt \parskip0pt \topsep3pt
\item \emph{Training and aggregation.} On $(\mathsf{Train},\Dj)$ from every
  client, run the local training of \cref{sec:split} on each $\Dj$ in both
  servable arrangements and store the two shared heads $\thstar_{A}$ and
  $\thstar_{B}$ of \cref{eq:aggregate}. Send
  $(\mathsf{Done},\mathsf{train})$ to every party.
\item \emph{Selection.} On $(\mathsf{Select})$ from $S$, evaluate the estimator
  of \cref{eq:estimator} on each of the two stored heads. Send the index
  $a^{\star}\in\{A,B\}$ of the larger and
  $(\mathsf{Done},\mathsf{select})$ to every party.
\item \emph{Serving.} On $(\mathsf{Query},x)$ from client $P_j$, return $\bot$
  to $P_j$ if $q_j\ge Q$. Otherwise set $q_j\leftarrow q_j+1$, return
  $\hat y=\arg\max_c\,(\thstar_{a^{\star}}\,\featx^{a^{\star}}_j(x))_c$ to
  $P_j$ alone, and send $(\mathsf{Done},\mathsf{query},j)$ to $S$. Here
  $\featx^{A}_j=\featx$ is the public backbone and
  $\featx^{B}_j=\featx_j$ is the backbone composed with $P_j$'s adapter.
\end{enumerate}

\noindent\textbf{\emph{Leakage.}} $\Fhe$ reveals $\Leak$ to the adversary,
consisting of the public parameters, the per-client sample counts $\nj$, the
announced index $a^{\star}$, and the fact that each $\mathsf{Done}$ message was
sent.
\end{functionality}

The leakage includes the sample counts because HE-OFT uses them as public
scalars. We prove security relative to $\Fhe$ with this leakage.

A protocol that realizes $\Fhe$ inherits three of its properties. No party
holds either shared head. The federation as a whole learns one index. Each client
learns one label for each query it submits, up to its allowance $Q$, together
with the leakage $\Leak$.

\subsection{Security Against Semi-Honest Parties}
\label{sec:semihonest}

\begin{definition}[Semi-honest security]
\label{def:semihonest}
Let $\mathcal{A}$ corrupt a set of parties that may contain the server and at
most $\tc-1$ clients. Write
$\mathsf{view}_{\mathcal{A}}^{\Pi}(\Dset)$ for the corrupted parties' joint view
of a protocol run on datasets $\Dset$. The protocol $\Pi$ securely realizes
$\Fhe$ with leakage $\Leak$ if there is a probabilistic polynomial-time
simulator $\mathcal{S}$ such that, for every $\Dset$,
\[
  \mathsf{view}_{\mathcal{A}}^{\Pi}(\Dset)
  \;\stackrel{c}{\approx}\;
  \mathcal{S}\big(\Leak,\ \{\Dj\}_{j\ \mathrm{corrupt}},\ \mathsf{out}_{\mathcal{A}}\big),
\]
where $\mathsf{out}_{\mathcal{A}}$ collects the answers $\Fhe$ returns to the
corrupted clients.
\end{definition}

The simulator receives the answers because HE-OFT delivers them. A
definition that withheld them would be unachievable, and the resulting theorem
would not apply to a system that answers queries.

\begin{theorem}[Semi-honest security]
\label{thm:semihonest}
Assume the $\tc$-out-of-$\Nc$ threshold CKKS scheme is secure under chosen
plaintext attack (IND-CPA) against an
adversary holding $\tc-1$ key shares. Assume also that a key switch smudges above the
noise of the ciphertext it switches. Then HE-OFT securely realizes $\Fhe$
with leakage $\Leak$ against any semi-honest adversary corrupting the server
and at most $\tc-1$ clients, in the sense of
\cref{def:semihonest}.
\end{theorem}

\begin{proof}[Proof]

We write $T$ for the corrupted clients, with $|T|\le\tc-1$, and $H$ for the
honest ones. We treat the case in which the adversary also corrupts the server.
Without the server, the adversary's view is a part of this one, and the
simulator discards the rest.

\emph{The assumptions, stated precisely.} We read the IND-CPA assumption in the
following game. The challenger and the adversary run the distributed key
generation, with the adversary playing the clients in $T$ and the challenger the
clients in $H$. The challenger publishes the collective public key together with
the evaluation keys, which comprise the relinearization, rotation and
bootstrapping keys. The adversary submits pairs of plaintexts and receives an
encryption of one plaintext of each pair, selected by a hidden bit. We write
$\varepsilon_{\mathrm{cpa}}$ for the advantage of a polynomial-time adversary
with one challenge pair. IND-CPA in this game is a circular-security
assumption. It is standard for CKKS and is not implied by RLWE. The challenger also answers challenges in the form a refresh outputs. Such a
ciphertext is an encryption under the collective secret key whose second
component is a fresh public random element. Indistinguishability in this form
follows from RLWE, and we include it in $\varepsilon_{\mathrm{cpa}}$. The same
game for a single key, without distributed key generation, defines IND-CPA under a client's own public key, and
we write $\varepsilon'_{\mathrm{cpa}}$ for its advantage.

After key generation, HE-OFT runs three interactive subprotocols of the
multiparty scheme. They are the collective refresh, the collective decryption,
and the collective key switch to a client's public
key~\cite{mouchet2021multiparty,mouchet2023threshold}.
In each, every participating client adds a smudging term to its share. Each
subprotocol has a simulator that produces the honest clients' shares from the
subprotocol's output and from the corrupted clients' key shares. The smudging
assumption of the theorem states that this simulation is close.

Formally, let $e$ be the noise of the ciphertext a subprotocol processes and $E$ the smudging
term. We require the statistical distance between $E$ and $E+e$ to be at most
$\varepsilon_{\mathrm{sm}}$ for every ciphertext the protocol processes. For
Gaussian smudging of standard deviation $\sigma$ on each of $n$ coefficients and
noise with $\|e\|_\infty\le B$, this distance is at most $nB/\sigma$, and for discrete Gaussians whose parameter
exceeds the smoothing parameter of $\mathbb{Z}$ the same bound holds up to a
negligible additive term. In the refresh, the honest masks play the role of $E$,
and $e$ is the whole decrypted value, the plaintext together with its noise. We
read the smudging hypothesis as covering this case. The implementation sizes
the masks for $128$ bits of statistical security.

We write $K\varepsilon_{\mathrm{sub}}$ for the distance between the real runs of
all $K$ subprotocols and their joint simulation. It is at most
$K\varepsilon_{\mathrm{sm}}$ plus a term negligible under the ring learning with
errors (RLWE) assumption on which the scheme rests. The RLWE term covers all runs
jointly, since every run reuses the honest key shares, and we take this joint
simulation from~\cite{mouchet2021multiparty,mouchet2023threshold}.

\emph{Counting.} Let $K$ be the number of subprotocol runs. They are the two
refreshes of the encrypted reciprocal, the decryption of $a^{\star}$, and one key
switch for each answered query, at most $\Nc Q$ in all. Let $M$ be the number of
ciphertexts under the collective key that honest clients encrypt or that the
refreshes output. They are the training uploads, the selection uploads, the
features of the honest clients' queries, and the two refreshed ciphertexts. \Cref{alg:select}
pads each client's selection upload to a number of ciphertexts that the public
parameters fix, and it masks the final comparison to one slot before the
decryption. The public parameters and the number of queries therefore fix $K$
and $M$, and both are polynomial in $\lambda$.

\emph{The simulator.} $\mathcal{S}$ receives $\Leak$, the datasets of the clients
in $T$, and the answers $\mathsf{out}_{\mathcal{A}}$.
\begin{enumerate}
\item \emph{Key generation.} $\mathcal{S}$ samples key shares for the clients in
  $H$ and runs the distributed key generation with the adversary. It also
  samples a key pair for each honest client, under which that client receives its
  answers.
\item \emph{Training.} For each ciphertext an honest client uploads,
  $\mathcal{S}$ sends an encryption of zero under the collective key. The number
  of these ciphertexts follows from $\Cc$, the feature dimension and the ring
  degree. $\mathcal{S}$ computes the server's additions and products itself. At
  each of the two refreshes, $\mathcal{S}$ takes an encryption of zero in the form a
  refresh outputs as the refreshed ciphertext and produces the honest shares with the refresh simulator.
\item \emph{Selection.} $\mathcal{S}$ sends encryptions of zero in place of the
  honest clients' encrypted features and encrypted one-hot labels. It evaluates
  the selection circuit, whose level restorations use the public bootstrapping
  keys and need no client. It reads $a^{\star}$ from $\Leak$ and produces the
  honest decryption shares with the decryption simulator on output $a^{\star}$.
\item \emph{Serving.} A corrupted client produces the encrypted features of its
  own queries. $\mathcal{S}$ evaluates the serving circuit on them and produces
  the honest key-switching shares with the key-switch simulator, on the output
  that decrypts under the querier's key to the label in
  $\mathsf{out}_{\mathcal{A}}$. For a query of an honest client $j$, which
  $\mathcal{S}$ learns from the signal $(\mathsf{Done},\mathsf{query},j)$ in
  $\Leak$, it sends an encryption of zero as the encrypted features, evaluates
  the circuit, and runs the key-switch simulator on an encryption of zero under
  the public key of client $j$.
\end{enumerate}
$\mathcal{S}$ runs in polynomial time. Its output depends on the honest clients'
data only through $\Leak$ and $\mathsf{out}_{\mathcal{A}}$.

\emph{Hybrids.} We pass from the real view $\mathsf{H}_0$ to the output of
$\mathcal{S}$ in three steps. The order matters. A reduction to IND-CPA does not
hold the honest key shares, and it cannot compute real shares for ciphertexts that
depend on its challenge. The step that removes those shares therefore comes
first.

In $\mathsf{H}_1$, every subprotocol run is replaced by its simulator, applied to
the real output of that run. After this step the honest key shares enter the view
only through key generation. The joint replacement changes the view by at most
$K\varepsilon_{\mathrm{sub}}$.

In $\mathsf{H}_2$, the output of every key switch to an honest client is replaced
by an encryption of zero under that client's public key. The honest client's
secret key appears nowhere else in $\mathsf{H}_1$. A distinguisher between the two
hybrids thus yields an adversary against IND-CPA under that key, and the at most
$\Nc Q$ replacements cost at most $\Nc Q\,\varepsilon'_{\mathrm{cpa}}$.

In $\mathsf{H}_3$, the $M$ ciphertexts under the collective key that honest clients
encrypt or that the refreshes output are replaced, one at a time, by encryptions
of zero. The reduction for one replacement runs key generation through the
IND-CPA game and places the challenge at the position being replaced. It produces
every other message itself. It holds every dataset, and it therefore computes the
plaintexts of the other ciphertexts and the outputs that the subprotocol
simulators need, namely $a^{\star}$, the labels of the corrupted clients' queries
and the refreshed values. It uses no honest key share, by $\mathsf{H}_1$. The
server's computations are public functions of ciphertexts and of the evaluation
keys. The $M$ replacements cost at most $M\varepsilon_{\mathrm{cpa}}$.

$\mathsf{H}_3$ is the output of $\mathcal{S}$. Every honest ciphertext is an
encryption of zero, the server's values are computed from these, and the
subprotocol simulators receive $a^{\star}$, the answers in
$\mathsf{out}_{\mathcal{A}}$, and encryptions of zero. For every polynomial-time
distinguisher, the gap between the real view and the simulated one is therefore at
most
\[
  K\varepsilon_{\mathrm{sub}}+\Nc Q\,\varepsilon'_{\mathrm{cpa}}+M\varepsilon_{\mathrm{cpa}},
\]
which is negligible when $\varepsilon_{\mathrm{sm}}$,
$\varepsilon_{\mathrm{cpa}}$ and $\varepsilon'_{\mathrm{cpa}}$ are. The bound
says nothing when $\varepsilon_{\mathrm{sm}}$ is not negligible, and
\cref{sec:cpad} shows that our measured configuration is such a case.

\end{proof}

\subsection{Input Privacy Against Malicious Clients}
\label{sec:malicious}

\Cref{thm:semihonest} assumes that every party follows the protocol. We now let
up to $\tc-1$ clients deviate arbitrarily and bound what they learn about the data
of a client that follows the protocol. We do not claim correctness against such
clients.

\begin{definition}[Content game]
\label{def:contentgame}
The adversary $\mathcal{A}$ corrupts $\tc-1$ clients. The server is honest. At least one client $h$ stays honest.
\begin{enumerate}
\item $\mathcal{A}$ outputs two datasets $D_0$ and $D_1$ with $|D_0|=|D_1|$ that lead to the same selected index $a^{\star}$.
\item The challenger samples $b\in\{0,1\}$ and gives $D_b$ to client $h$.
\item HE-OFT runs. $\mathcal{A}$ may deviate arbitrarily in the values its
  clients upload, in the shares they contribute, and in the queries they ask.
\item $\mathcal{A}$ outputs $b'$. Its advantage is $|\Pr[b'=b]-\tfrac12|$.
\end{enumerate}
\end{definition}

The equal-size restriction is necessary. The per-client sample counts are
public. An adversary free to choose datasets of different sizes determines $b$
from the counts alone.

\begin{theorem}[Input privacy]
\label{thm:malicious}
Assume the conditions of \cref{thm:semihonest}, with key generation run over
committed shares. Then the advantage of
$\mathcal{A}$ in \cref{def:contentgame} is at most
$\negl(\lambda)+\delta(Q_{\mathrm{tot}})$, where
$Q_{\mathrm{tot}}=(\tc-1)Q$ and $\delta$ is the advantage available from
$Q_{\mathrm{tot}}$ label queries to the served head.
\end{theorem}

\begin{proof}[Proof]

We write $T$ for the $\tc-1$ deviating clients and $h$ for the honest client
that holds $D_b$. The server and every client outside $T$ follow the protocol.

\emph{Hypotheses on deviating messages.} A deviating client can send messages that
no honest client produces. Three readings fix what the theorem covers. First, we
read key generation over committed shares as using extractable commitments, so
that a reduction recovers each deviating client's key shares. Each deviating
client fixes its message of each round of key generation before it sees an
honest message of that round. We also read the threshold sharing as verifiable,
so that the shares a deviating client deals determine one key share. Every upload
and query of a deviating client is likewise an encryption of a plaintext that the
client knows and that a reduction can extract. The
collective public key and the evaluation keys then contain the honest clients'
contributions as independent samples, and the IND-CPA game of the proof of
\cref{thm:semihonest} applies with the deviating clients' shares chosen by the
adversary.

Second, a malformed upload can carry more noise than an honest
encryption, and a ciphertext computed from it can then violate the smudging
condition when it is refreshed, decrypted or switched. We read the smudging
condition of the theorem as holding for every ciphertext the protocol processes,
including those computed from deviating uploads. Third, a malformed query can
encrypt a feature vector outside the domain for which the comparisons of the
argmax are parameterized. The comparisons then return a polynomial in the logits
instead of a sign, and the returned value is not a label. We read $\delta$ as
covering queries whose features lie in that domain. All three readings hold if
every upload, query and key-generation message of a deviating client carries a
zero-knowledge proof of knowledge of its plaintext or key share, of bounded noise,
and of bounded features. HE-OFT does not include such proofs, and the theorem
holds for deviating clients whose messages satisfy the three readings.

\emph{The view.} The view of $\mathcal{A}$ consists of the key-generation
transcript, the honest clients' shares in the refreshes, decryptions and key
switches in which its clients take part, the ciphertexts the server presents in
those runs, and the answers to its own queries. The honest server presents a
ciphertext only when the protocol prescribes it, and that ciphertext is the
prescribed function of all uploads and of the query. Each upload and query of
a deviating client is a function of the adversary's view at the time it is sent.
By the first reading, it is an encryption of a plaintext that the reduction
extracts.

\emph{Deviating shares.} A deviating client may contribute any value as its share
in a subprotocol run. An honest share depends only on the presented ciphertext and
on the honest client's key share, and it does not depend on any other share. A
deviating share therefore changes what the combined result decrypts to and leaves
the honest shares unchanged. The adversary computes the combined result from the
honest shares and its own, and deviating gives it no information beyond the
honest shares.

\emph{Hybrids.} Let $\mathsf{G}_0(b)$ be the game of \cref{def:contentgame} with
challenge bit $b$. We follow the three steps of the proof of
\cref{thm:semihonest}.

In $\mathsf{G}_1(b)$, the honest shares in every subprotocol run are produced by
the simulator of that subprotocol from the honest clients' part of the output.
The reduction computes that part from the plaintext of the presented ciphertext
and from the deviating clients' key shares, which it extracts. By the second
reading, the smudging condition holds in every run, and the $K$ replacements cost
at most $K\varepsilon_{\mathrm{sub}}$.

In $\mathsf{G}_2(b)$, the output of every key switch to an honest client is
replaced by an encryption of zero under that client's key, at a cost of at most
$\Nc Q\,\varepsilon'_{\mathrm{cpa}}$.

In $\mathsf{G}_3(b)$, every ciphertext under the collective key that an honest
client encrypts or that a refresh outputs is replaced by an encryption of zero, at
a cost of at most $M\varepsilon_{\mathrm{cpa}}$. The reduction needs the
plaintexts of the presented ciphertexts. It holds every honest dataset, it
obtains the plaintexts of the deviating clients' uploads and queries by the first
reading, and it computes the presented plaintexts from these. A deviating refresh
share shifts the refreshed value by an amount the reduction computes from the
extracted key shares, and the replacement keeps that shift.

In $\mathsf{G}_3(b)$, the view depends on $b$ only through $a^{\star}$ and through
the answers to the deviating clients' queries. \Cref{def:contentgame} fixes
$a^{\star}$ to the same value for $b=0$ and $b=1$. The honest server enforces the
allowance, and the deviating clients therefore receive at most
$Q_{\mathrm{tot}}=(\tc-1)Q$ answers. By the third reading each answer is a label
of the served head on a query the adversary chooses, possibly adaptively. The
distance between $\mathsf{G}_3(0)$ and $\mathsf{G}_3(1)$ therefore gives an
advantage of at most $\delta(Q_{\mathrm{tot}})$, the advantage of an adversary
that chooses the deviating clients' uploads and their shares in the refreshes,
and then asks $Q_{\mathrm{tot}}$ label queries to the resulting head. We measure
$\delta$ as $|\Pr[b'=b]-\tfrac12|$, as in \cref{def:contentgame}. Collecting the three
steps, the advantage of $\mathcal{A}$ is at most
\[
  K\varepsilon_{\mathrm{sub}}+\Nc Q\,\varepsilon'_{\mathrm{cpa}}
  +M\varepsilon_{\mathrm{cpa}}+\delta(Q_{\mathrm{tot}})
  \;=\;\negl(\lambda)+\delta(Q_{\mathrm{tot}}).
\]
Here $\negl(\lambda)$ also absorbs the binding and extraction errors of the
commitments and of the plaintext extraction.

\end{proof}

The bound grows with the corruption threshold. A coalition of $\tc-1$ clients
holds $(\tc-1)Q$ queries regardless of the size of the federation.
We set $Q$ in \cref{sec:exp-leak} from the cost of extraction, between three and
five queries per parameter of the head for a copy of fidelity $0.90$. Against the head
in plaintext, the likelihood-ratio membership attack stays close to a random
guess (\cref{tab:surfaces}).

\subsection{The Honest-Server Assumption}
\label{sec:malicious-ext}

\Cref{thm:malicious} assumes an honest server. We show that this assumption is
necessary for every protocol in which the clients decrypt ciphertexts that the
server presents. A malicious server can present a ciphertext of its own choosing,
such as a row of $\thstar$, in place of the output of \cref{alg:serve}. An honest client would have to refuse to decrypt it. By \cref{prop:nogate}, the client's decision cannot depend on the plaintext.

\begin{proposition}[No plaintext gate from a ciphertext]
\label{prop:nogate}
Let an honest client decide, by a polynomial-time predicate $D$ on its view,
whether to contribute its decryption share for a presented ciphertext. Let the
rest of that view be independent of the underlying plaintext. If the scheme is
IND-CPA against an adversary holding $\tc-1$ key shares, then for every pair of
plaintexts $m_0$ and $m_1$,
\[
  \big|\Pr[D(\mathsf{Enc}(m_0))=1]-\Pr[D(\mathsf{Enc}(m_1))=1]\big|
  \le \negl(\lambda).
\]
\end{proposition}

\begin{proof}
A gap between the two probabilities gives an IND-CPA distinguisher that holds
one key share. The scheme is assumed secure against any adversary holding up to
$\tc-1$ key shares, which includes one share, and such a distinguisher therefore
contradicts that assumption.
\end{proof}

Correctness requires an honest client to contribute its share for the
prescribed ciphertext. By \cref{prop:nogate}, it then contributes its share for
the ciphertext of a malicious server with the same probability, up to a
negligible difference. Hence no protocol in which the clients decrypt ciphertexts
that the server presents realizes $\Fhe$ against a malicious server.

Two mechanisms would remove the assumption.
The first verifies where a ciphertext presented for decryption came from. The
serving circuit is a deterministic function of the encrypted head, the encrypted
query, the evaluation keys, and the public parameters. Its level restoration is
deterministic as well, because HE-OFT restores levels with collectively generated
bootstrapping keys and not by collective refresh. A checker that holds these four
inputs can therefore recompute the prescribed output ciphertext and compare it
bit for bit with the presented one. Restoring levels by bootstrapping thus saves
$510$\,MiB of traffic per query and also makes the circuit verifiable. The
checker must recompute every query. A malicious server can present a whole row of
the head in one request, and at four classes the whole head fits in one
ciphertext. A single unchecked request can thus leak the head.

The first mechanism has two costs that HE-OFT as specified does not
incur. The checker must hold the encrypted head and the querier's encrypted
features, and no quorum member holds either. Broadcasting both to the quorum
discloses no plaintext, because both are ciphertexts. The broadcast adds traffic to every
checked query. It also reveals to the quorum which client submitted the query,
and we count this as leakage in \cref{sec:exp-leak}.

The second mechanism certifies which function the server
evaluated, through a zero-knowledge proof of correct evaluation or verifiable
homomorphic encryption~\cite{viand2023verifiable}. It incurs neither of these two costs.
\Cref{sec:scope} lists both mechanisms as limitations.

\subsection{Leakage to a Client Coalition}
\label{sec:c2c}

\Cref{thm:malicious} bounds the advantage of a deviating coalition by
$\negl(\lambda)+\delta(Q_{\mathrm{tot}})$ and leaves $\delta$ unspecified. We now
bound $\delta$. A client coalition has more access than any other
adversary in this model, because its members hold key shares, run the backbone
themselves, and receive answers.

A client outside the coalition contributes four things to what the
coalition sees: (i) the public sample count $\nj$, which HE-OFT uses as a
public scalar, (ii) the ciphertexts it uploads, which are computationally
independent of what they encrypt, (iii) its share of the announced index
$a^{\star}$ and of the phase signals, and (iv) its part of every answer the
coalition is given. It contributes no per-class total, because the server
inverts those under encryption and never decrypts them (\cref{sec:split}), and no adapter, because
no adapter aggregate exists.

The allowance bounds how many answers a coalition receives, and a
membership attack needs few of them. Membership inference from predicted labels
is an established attack on a classifier that answers
queries~\cite{choquettechoo2021labelonly}. An allowance set against
extraction leaves room for this attack. We therefore bound what the answers reveal about a
record by what the head itself reveals, independently of the allowance.

\begin{proposition}[The query channel is no stronger than the head]
\label{prop:c2c}
Assume the conditions of \cref{thm:malicious} and the readings of its proof. Let
$\mathcal{A}$ corrupt at most $\tc-1$ clients in \cref{def:contentgame}, with
the server honest, and let it ask at most $Q_{\mathrm{tot}}=(\tc-1)Q$ queries.
Then there is an adversary $\mathcal{B}$ that receives both shared heads in
plaintext together with $\Leak$, asks no query, and satisfies
\[
  \mathsf{Adv}(\mathcal{A})\;\le\;\mathsf{Adv}(\mathcal{B})+\negl(\lambda).
\]
\end{proposition}

\begin{proof}
In $\mathsf{G}_3$ of the proof of \cref{thm:malicious}, the view of $\mathcal{A}$
consists of a part that is independent of $b$ and of the answers to its own
queries. $\mathcal{B}$ produces the first part as $\mathsf{G}_3$ does. It extracts
the deviating clients' key shares and contributions by the first reading of that
proof, submits the contributions, and receives both heads computed from them. It
reads $a^{\star}$ from $\Leak$.

For each query of $\mathcal{A}$, $\mathcal{B}$ extracts the plaintext features by
the same reading. A deviating client need not hold an adapter, and the answer
depends on the query only through these features. $\mathcal{B}$ holds the head
$\thstar_{a^{\star}}$ in plaintext and therefore computes the value the serving
circuit returns on the features, which is the answer $\mathcal{A}$ receives.

$\mathcal{B}$ runs $\mathcal{A}$, answers each of its queries in this way, and
outputs the bit that $\mathcal{A}$ outputs. The view $\mathcal{B}$ produces
differs from the real one by the negligible term of the proof of
\cref{thm:malicious}.
\end{proof}

Write $\delta_{\mathrm{wb}}$ for the advantage of the best
polynomial-time adversary that holds both shared heads. By \cref{prop:c2c}, $\delta(Q_{\mathrm{tot}})\le\delta_{\mathrm{wb}}$ at every budget. The
advantage of a coalition against an honest client is therefore at most
$\negl(\lambda)+\delta_{\mathrm{wb}}$, which does not depend on the allowance. The
allowance determines how many queries a coalition spends to approach
$\delta_{\mathrm{wb}}$, and \cref{sec:exp-leak} measures that number.

Raising either the threshold $\tc$ or the allowance $Q$ raises
$Q_{\mathrm{tot}}=(\tc-1)Q$ and lets a coalition come closer to
$\delta_{\mathrm{wb}}$. Neither parameter changes $\delta_{\mathrm{wb}}$ itself.

We estimate $\delta_{\mathrm{wb}}$ with the adversary of
\cref{prop:c2c}. The attacker receives both shared heads in plaintext and runs
the likelihood-ratio membership attack of Carlini
et al.~\cite{carlini2022membership} against them. We build $64$ shadow
federations per cell, each consisting of $\Nc$ clients that train on half of
their own examples and merge by \cref{eq:aggregate}. We test $2000$
candidates over three tasks, three seeds and both arrangements. Averaged over
the three seeds, the attack reaches an area under the curve between $0.49$ and
$0.53$. Its averaged true-positive rate at a false-positive rate of $0.01$ lies
between $0.008$ and $0.023$. A
random guess has an area of $0.5$ and a true-positive rate equal to the
false-positive rate, and the measured values lie close to both.

No party in HE-OFT holds either head. This experiment thus evaluates a
setting that arises only if the encryption is broken. By \cref{prop:c2c}, the
query channel is at most as informative as this setting.

Extracting a copy of the head before the attack does not raise the
advantage. We fit a copy of the head to $2\times10^{5}$ label-only answers, the
largest budget of \cref{sec:exp-leak}, and run the same attack against the copy.
Averaged over three seeds, the area under the curve lies between $0.48$ and
$0.50$, and the true-positive
rate at a false-positive rate of $0.01$ between $0.009$ and $0.017$. The result agrees with \cref{prop:c2c}.

The one-shot label-only (OSLO) membership attack does not raise the
advantage either. OSLO crafts one query per candidate instead of reading a statistic off a model. On CIFAR-100, Peng et al. report a true-positive rate of $6.7$ per cent at a false-positive rate of $0.001$, where the previous best label-only attack reaches $0.3$~\cite{peng2024oslo}. The serving
interface of HE-OFT exposes more to this attack than a standard prediction interface. A client uploads its own encrypted feature vector. Under \cref{thm:malicious}, a coalition may choose its queries freely. A crafted feature
vector therefore reaches the head unchanged. Over six cells, averaged over three
seeds, the attack reaches an area under the curve between $0.47$ and $0.52$.

The calibration step of the attack explains this result. A coalition
chooses its perturbation on its own data and knows which examples
each surrogate trained on. Even there, the best separation it reaches is at most
$0.065$, and below $0.01$ in most cells. The crafted query transfers to the served
head, but the head shows no difference in margin between members and non-members
for the query to exploit.

One further measurement is confounded. Predicting membership from whether the head classifies an example correctly reaches an area
under the curve of $0.72$ on Banking77. Our holdout contains at least one example
of every class a client holds. Held-out examples are therefore richer in rare classes than
training examples, and the rule separates the two sets for a reason unrelated to
membership.
The likelihood-ratio attack calibrates each example against its own shadow
distributions and does not inherit the confound.

The coverage-weighted merge admits a more targeted attack. A class that one client alone holds takes its row from that client's displacement, with no averaging over the federation. For a linear map with a bias, the ratio of the row displacement to the bias displacement
is a weighted mean of that client's features~\cite{phong2018privacy}. We compute
that ratio for all $570$ class rows of the three tasks. Where one client holds the class, the ratio matches the closest example of its own class better than the closest example of another class by $0.012$ in cosine. Where six or more clients hold it, the margin is $0.006$. The absolute cosine is $0.16$, and a random direction reaches $0.029$. The row points into the subspace the features occupy. It identifies
neither a class nor a record.

Three other attack families require structure that HE-OFT does not
provide. The first is reconstruction from weights. The strongest published attack in our setting recovers
training images from heads trained on frozen $768$-dimensional embeddings. Its
authors report that it fails on linear probes and name linear models as a
mitigation~\cite{oz2024reconstructing}. Their own governing ratio is the
parameter count over the unknowns, which for a linear head reduces to the class
count over the client's sample count. Ours is about $0.01$, and the smallest
ratio at which they report any recovery is $0.5$.

Attribute inference is the second. Against a model trained on features of
this kind, the black-box attack performs at the base rate. The attack reaches $0.28$ against a base
rate of $0.28$ on one dataset and $0.05$ against $0.06$ on
another~\cite{jayaraman2022attribute}. The variant that outperforms imputation
reads individual hidden-layer activations, and a single linear map has none.

Gradient inversion is the third, and it does not apply here. Every attack
in that line needs gradients with respect to parameters that see the input, or a
model the server chose~\cite{zhu2019deep,geiping2020inverting}. The backbone here
is public and fixed, every client runs the same one, and the only uploaded object
is a trained displacement of a map placed after the last nonlinearity.

The membership measurements have two limits. A thousand members resolve a
true-positive rate no finer than $0.001$. At a false-positive
rate of $0.001$, each measurement rests on two to four examples. The shadow
federations vary who trains on what and hold the coverage pattern fixed. They
thus say nothing about how leakage changes with the partition.

The membership attack of Carlini et
al.~\cite{carlini2022membership} remains close to a random guess after successful
extraction and when we disclose a head in plaintext. The head is a
linear map fitted for two hundred steps over thousands of examples per class.

\subsection{Approximate Decryption and CPA-D Security}
\label{sec:cpad}

HE-OFT performs two kinds of decryption. The selection step decrypts
one scalar, and the quorum key-switches each answer to the querier. Both are intended outputs. In an exact scheme, a decryption returns only the
message. An intended output hence reveals nothing beyond itself. CKKS is approximate. A decryption
returns the message together with an error term, and that term is a function of
the secret key and of the values the circuit consumed. Li and Micciancio show
that an adversary which observes decryption results recovers the key that
produced them~\cite{limicciancio2021approximate}.

Their attack does not need the message. It reads the ring element
$c_0+c_1 s$ out of the decoded output and inverts it. An adversary can therefore mount the attack without predicting the output of the circuit. In HE-OFT a single decryption does not give this ring element. Every
decryption follows a collective key switch, and each switching share adds
smudging noise. The querier therefore obtains $c_0+c_1s$ only up to that noise,
and no party ever decrypts $\mathsf{Enc}(\thstar)$. Repeated queries can average
the noise away, and we analyze that attack below.

The published threshold attack does not apply. It identifies ciphertexts whose noise is exactly zero and reads a linear
equation in the key off each one~\cite{checri2024threshold}. Detecting a
zero-noise ciphertext requires a rounding step with observable success,
which exact schemes have and CKKS does not. The statement of Checri et al. about threshold CKKS
covers the case where no smudging is applied.

Our implementation does not meet a condition on the smudging parameter. The key-switch simulator of the multiparty scheme
we build on is correct when the smudging noise dominates the noise of the
ciphertext being switched. This condition lets the simulator produce an honest client's share
without knowing that noise~\cite{mouchet2021multiparty}. That scheme sets the
requirement at $2^{\lambda/2}$ times the ciphertext noise. Our implementation
smudges at eight times the fresh-encryption noise, the value the library uses in its own tests. The serving circuit leaves a ciphertext noise several
orders of magnitude larger than that. The simulation step of
\cref{thm:semihonest} therefore rests on a hypothesis our measured configuration
does not satisfy. When the hypothesis holds, the step hides the circuit noise,
which depends on the honest clients' displacements. Whether an adversary can exploit that noise is
an open question, and the authors of the scheme record it as open as well.

Meeting the stated requirement takes more than raising one constant.
At a plaintext scale of $2^{45}$, the required smudging exceeds the scale itself
and would erase the answer. A configuration that
satisfies the condition needs a larger scale and modulus, and every cost in
\cref{sec:exp-cost} would move with them. Prior work reports the same
conflict. Threshold CKKS is insecure under this
notion when no smudging is applied, and flooding sized to the rule is expected
to cost the precision of the decrypted result~\cite{checri2024threshold}. We
report the configuration we measured together with this gap.

Repeated queries turn this gap into the key-recovery attack known for
approximate homomorphic encryption~\cite{limicciancio2021approximate}. The
number of repetitions it needs determines which control a deployment requires. Write the
served ciphertext as $(c_0,c_1)$ under the collective key $s$. The querier
decrypts $c_0+c_1s+E$, where $E$ collects the honest parties' smudging terms. The
serving circuit is deterministic. A client that resubmits an identical query
receives the same $c_0$ and $c_1$ and a freshly drawn $E$. Averaging $k$
repetitions shrinks the averaged $E$. Once it falls below half a unit in every coefficient, rounding recovers $c_0+c_1s$ exactly. The noise of
the ciphertext is not a separate term. The quantity $c_0+c_1s$ is the plaintext
together with that noise. An adversary that holds $c_0$ and $c_1$ subtracts and
rounds, obtains $c_1s$, and recovers $s$ by one ring inversion. The server holds
$c_0$ and $c_1$. Under \cref{def:semihonest}, the adversary corrupts the server together with
$\tc-1$ clients. This view falls inside the model. HE-OFT also permits a
client to resubmit a query, and the adversary hence remains semi-honest.

The cost is the number of repetitions that drives the averaged error
below half a unit in each of the $2^{15}$ coefficients. For a resampled part of standard deviation $\sigma$, that number is on the order of $81\sigma^2$. At eight times
the fresh-encryption noise and a quorum of ten that is about $5\times10^{5}$
repetitions, which is the same order as the query budgets of
\cref{sec:exp-leak}. A duplicate-query check is therefore a required control. Re-randomizing a query defeats the check but not the control. A
re-randomized query is a different ciphertext. It thus yields one independent
sample and not a repetition, and the flooding is sized for independent samples. Sizing that flooding for $q$ decryptions costs
half a bit for each doubling of $q$, and the budget is refreshed by rekeying once
it is spent~\cite{limicciancio2022dp}.

Raising the smudging is the second control. The values HE-OFT decrypts are one index and
one label. The precision the answer needs is a few bits against a plaintext
scale of $2^{45}$. The number of repetitions grows with the square of the
smudging noise. Moving the smudging from eight times the fresh-encryption noise
to $2^{25}$ raises it far beyond any query budget and still leaves twenty bits of
the scale. It changes neither latency nor traffic, and
it does not change the accuracy of the served answers, because HE-OFT returns
only the index of the largest score. It does not meet the $2^{\lambda/2}$ rule.
The served ciphertext enters the key switch with noise of standard deviation
$2^{12.9}$ at four classes and $2^{16.8}$ at a hundred. The smudging of
$2^{25}$ exceeds that noise by $8$ to $12$ bits. Over $150$ key switches, every one
returned the plaintext label, and the decoded index was never more than $0.13$
from the integer. We recommend this configuration.

One deviation in key generation lies outside the theorems. Collective key
generation has each party publish a share of the public key, and the shares are
summed. A rushing party publishes its share after seeing the others.
It can then subtract their contribution and obtain a public key whose secret it alone
holds. Every ciphertext
uploaded under that key is then readable by that party. The honest parties detect the deviation only when their decryption shares fail to reconstruct, and by then the uploads have been read. The standard remedy is to commit the
shares before opening them, which costs one round and is independent of the ring
parameters. Under \cref{thm:malicious}, a coalition may deviate in the shares it
contributes. The proof covers that deviation for decryption shares, where a
wrong share changes only the value the querier decrypts. That argument does not cover key-generation
shares. The theorem hence applies only to a deployment that includes the
commitment round.
\section{Experiments}
\label{sec:experiments}

We measure what HE-OFT learns under encryption (\cref{sec:exp-split}), how it
combines heterogeneous clients (\cref{sec:exp-select}) and how it compares with
one-shot federated learning (\cref{sec:exp-peers}). We then measure what the
encryption costs against encrypted systems (\cref{sec:exp-cost}) and what a
querier can learn at test time (\cref{sec:exp-leak}).

\begin{figure*}[t]
\centering
\includegraphics{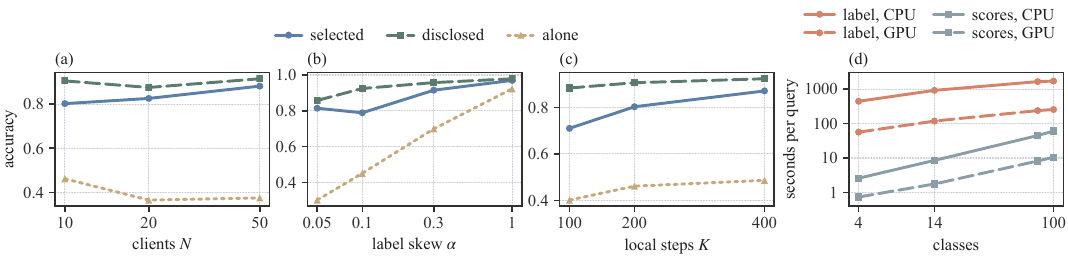}
\caption{HE-OFT along four axes. Panels (a) to (c) vary the federation
size, the label skew and the local trajectory on DBpedia around $\Nc=10$,
$\alpha=0.1$ and $K=200$. Panel (d) gives the test time of one query
at $\Nc=10$ against the number of classes, returning the label or the scores.
Times are single-core measurements or GPU projections (\cref{sec:exp-cost}). Panels (a) and (c) are one seed
and panel (b) is the mean over three.}
\label{fig:trends}
\end{figure*}

\subsection{Setup}
\label{sec:setup}

Each client trains a rank-$8$ low-rank adapter and a classifier head. The adapter's down-projection stays frozen at a shared public initialization~\cite{sun2024ffalora}.

We evaluate four text classification tasks of increasing label-space size
with a frozen RoBERTa-base encoder~\cite{liu2019roberta}: AG-News with $4$
classes~\cite{zhang2015charcnn}, TREC with $6$~\cite{li2002trec}, DBpedia with
$14$~\cite{zhang2015charcnn} and Banking77 with $77$~\cite{casanueva2020banking77}.
We evaluate CIFAR-100~\cite{krizhevsky2009cifar} with a frozen
ViT-B/16~\cite{dosovitskiy2021vit}. Each text task draws at most
$20{,}000$ training and $5{,}000$ test examples under the seed, and each vision
task at most $10{,}000$ and $2{,}000$. TREC and Banking77 are smaller than these
limits, and we use them whole. We partition the data across clients with a
Dirichlet distribution over labels with concentration $\alpha$, the standard
label-skew partition~\cite{hsu2019measuring,li2022niidbench}. It induces label-distribution
skew in the usual taxonomy~\cite{kairouz2021advances} and holds the feature
distribution fixed across clients. A smaller $\alpha$ leaves each client with a few
dominant classes and almost none of the rest. Unless stated otherwise, there are $\Nc=10$
clients, $\alpha=0.1$, the local trajectory is $200$ steps, and every number is the
mean over three seeds. Each client reserves at least one held-out example from
every class it holds (\cref{sec:selection}) and tops the held-out set
up to one tenth of its shard at random.

We write $\Asel$ for the accuracy of the selected servable arrangement
(\cref{sec:exp-select}), and $\Aloc$ for the accuracy a client obtains
alone. We write $\Adis$ for the accuracy of the disclosed model, in which
the adapter and the head are both aggregated and decrypted. We write $\Apool$ for the accuracy of one model
trained on the union of the clients' data.

The ratio $\Asel/\Adis$ is the share of the disclosed model's accuracy that
HE-OFT keeps without disclosing any model. The difference $\Apool-\Adis$ is what
the partition of the data costs even a disclosed model.

\subsection{One-Shot Encrypted Federated Learning}
\label{sec:exp-split}

Each client adapts its representation locally, and the federation shares only
the classifier head. The resulting model covers the whole label space, and
encryption does not degrade its accuracy. No prior framework keeps the contributions under
multiparty encryption while fine-tuning a pretrained transformer. POSEIDON~\cite{sav2021poseidon}  trains networks of a few layers from scratch, and
slytHErin~\cite{intoci2023slytherin} and CryptPEFT~\cite{xia2026cryptpeft} serve a model they do not train (\cref{tab:related}). We
therefore measure HE-OFT against references on its own tasks. We report both
servable arrangements, a client alone, the disclosed model, and pooled training
in \cref{tab:headline}.
The arrangement labeled \emph{shared head} gives every client the same head over
the bare public backbone. \emph{Personal adapter} gives
each client the shared head over its own locally adapted representation.

\begin{table}[t]
\centering
\caption{Accuracy on the global test set, three seeds, $\Nc=10$, $\alpha=0.1$. The servable columns are the two arrangements from which HE-OFT answers queries,
and bold marks the better one. The reference columns cannot be served under the
threat model of \cref{sec:threat}. Alone is one client's own model without
federation, disclosed is the merged model that HE-OFT does not build, and pooled
is centralized training on the union of the data. On one seed the
Dirichlet draw leaves some clients with fewer than twenty samples, too few to
train. The effective federation is thus seven on AG-News and nine on TREC.}
\label{tab:headline}
\footnotesize
\setlength{\tabcolsep}{4pt}
\begin{tabular}{lcccccc}
\toprule
 & & \multicolumn{2}{c}{servable} & & \multicolumn{2}{c}{reference} \\
\cmidrule(lr){3-4}\cmidrule(lr){6-7}
Task & $\Cc$ & shared & personal & alone & disclosed & pooled \\
 & & head & adapter & & & \\
\midrule
AG-News    & 4   & $\mathbf{0.649}$ & 0.479 & 0.475 & 0.739 & 0.921 \\
TREC       & 6   & $\mathbf{0.607}$ & 0.429 & 0.400 & 0.711 & 0.967 \\
DBpedia    & 14  & $\mathbf{0.789}$ & 0.754 & 0.451 & 0.925 & 0.988 \\
Banking77  & 77  & 0.206 & $\mathbf{0.686}$ & 0.249 & 0.760 & 0.923 \\
CIFAR-100  & 100 & 0.748 & $\mathbf{0.774}$ & 0.197 & 0.784 & 0.893 \\
\bottomrule
\end{tabular}
\end{table}

\Cref{tab:headline} supports three observations. First, the federation raises
accuracy on every task, and the gain is largest where a client's own data covers
the least of the label space.

Second, the better arrangement changes with the size of the label space. On the
small label spaces the shared head is better, and on the large ones the personal
adapter is. The AG-News margin rests on two of the three seeds, because the third
partition drops three clients. On Banking77 the shared head falls far below the
personal adapter, and missing coverage does not explain this drop. Every class row has three
or more holders on one of the three Banking77 seeds and at least one on all
three. The merge thus has evidence for every row. The bare public backbone
does not separate seventy-seven intents with a linear map, and each client's own
adapter does.

Third, the two arrangements fail under different conditions. The shared head
fails when the
public backbone alone does not separate the label space. The personal adapter
fails when each client's representation drifts far enough that a common head no
longer transfers.
The federation therefore selects between the two arrangements
(\cref{sec:exp-select}).

\begin{table}[t]
\centering
\caption{Standard deviation over the three seeds of \cref{tab:headline}, in
parentheses beside each mean. The spread is below $0.03$ everywhere except
AG-News and TREC. On AG-News the shared head reaches $0.809$, $0.736$ and
$0.402$, and the third partition drops three clients.}
\label{tab:headline-sd}
\footnotesize
\setlength{\tabcolsep}{2pt}
\begin{tabular}{lcccc}
\toprule
Task & shared head & personal adapter & alone & disclosed \\
\midrule
AG-News   & $0.649$ ($0.217$) & $0.479$ ($0.021$) & $0.475$ ($0.006$) & $0.739$ ($0.073$) \\
TREC      & $0.607$ ($0.074$) & $0.429$ ($0.068$) & $0.400$ ($0.051$) & $0.711$ ($0.043$) \\
DBpedia   & $0.789$ ($0.019$) & $0.754$ ($0.014$) & $0.450$ ($0.009$) & $0.925$ ($0.020$) \\
Banking77 & $0.206$ ($0.030$) & $0.686$ ($0.012$) & $0.249$ ($0.005$) & $0.760$ ($0.018$) \\
CIFAR-100 & $0.748$ ($0.010$) & $0.774$ ($0.003$) & $0.197$ ($0.009$) & $0.784$ ($0.004$) \\
\bottomrule
\end{tabular}
\end{table}

\Cref{tab:headline-sd} gives the standard deviation of every entry of
\cref{tab:headline}. Three of the five tasks vary by less than $0.03$ in every
arrangement. AG-News varies most because of its partition. One of its three Dirichlet draws leaves seven usable clients, and
the shared head reaches $0.402$ on that draw against $0.809$ and $0.736$ on the
other two. The estimator selects the personal adapter on that seed, and therefore
$\Asel$ exceeds the shared head's mean on AG-News.

\paragraph{The public initializer} $\thz$ is the state every client starts
from. Its head is the default random initialization of a linear layer under the shared public seed, and is therefore nonzero. A row that no client
updates keeps a random direction and still competes in the argmax. The count
above shows that no such row occurs in our runs, and the choice thus affects no
reported number. A deployment with uncovered classes should set those rows to
zero.

\paragraph{Clients per class row} Under the coverage-weighted merge, a row
decided by one client carries that client's displacement alone. We
counted the holders of every class row on the partitions the accuracy numbers were
measured on. On AG-News, TREC and DBpedia every class is held by three or more
clients on every seed. On Banking77 the same holds on one seed, and the other two
leave two classes with two holders and one class with a single holder. Every class is held by at least one client. Coverage therefore does not limit the shared head
on the large label spaces. \Cref{sec:c2c} uses the same counts to bound what a row
with a single holder could reveal.

\paragraph{Sensitivity} \Cref{fig:trends} varies the protocol axes around the
default on DBpedia. Accuracy falls with skew, from $0.970$ at
$\alpha=1.0$ to $0.789$ at $\alpha=0.1$. Longer local
training raises accuracy. Accuracy
rises with the size of the federation, from $0.803$ at $\Nc=10$ to $0.882$ at
$\Nc=50$. The federation-size panel is a single seed, and its $\Nc=10$ value therefore differs from
\cref{tab:headline}. \Cref{tab:sens} gives the underlying numbers.

The estimator chooses the shared head at $\alpha\in\{0.05,0.1\}$ and the personal adapter at
$\alpha\in\{0.3,1.0\}$, on every seed in each case. The crossover thus depends on
coverage as well as on the size of the label space. The personal adapter is preferable once each client sees enough of
the label space for its own representation to transfer, and the shared head is
preferable otherwise. The estimator finds this crossover from encrypted local
evidence, without a test set or a fitted threshold.

\begin{table}[t]
\centering
\caption{Sensitivity on DBpedia along the protocol axes, each varied around the
default of $\Nc=10$, $\alpha=0.1$ and $200$ local steps. The row reports the
accuracy of the arrangement the estimator selects. The skew row is the mean over
three seeds. The local-steps and client-count rows are a single seed, and their
default cells are that seed's value rather than the three-seed mean.}
\label{tab:sens}
\footnotesize
\begin{tabular}{lcccc}
\toprule
\multicolumn{5}{@{}l}{\emph{Label skew} $\alpha$} \\
 & $0.05$ & $0.10$ & $0.30$ & $1.00$ \\
selected arrangement & $0.814$ & $0.789$ & $0.915$ & $0.970$ \\
selected & shared & shared & personal & personal \\
\midrule
\multicolumn{5}{@{}l}{\emph{Clients} $\Nc$} \\
 & $10$ & $20$ & $50$ & \\
selected arrangement & $0.803$ & $0.826$ & $0.882$ & \\
\midrule
\multicolumn{5}{@{}l}{\emph{Local steps}} \\
 & $100$ & $200$ & $400$ & \\
selected arrangement & $0.710$ & $0.803$ & $0.871$ & \\
\bottomrule
\end{tabular}
\end{table}

 The encryption adds no loss of its own. The encrypted argmax agrees with the
plaintext maximum to $7.3\times10^{-5}$, far below the smallest scaled gap of
$0.0068$ between the top two logits we observed on real queries. The encrypted
argmax therefore selects the same class. On a recorded AG-News head, all $32$ encrypted answers
over the two arrangements equal the plaintext answers. The accuracy of
\cref{tab:headline} is therefore the accuracy HE-OFT serves under encryption.

HE-OFT keeps most of the disclosed model's accuracy without disclosing
any model. The retained share $\Asel/\Adis$ is $0.905$ on AG-News, $0.854$ on
TREC, $0.853$ on DBpedia, $0.902$ on Banking77 and $0.963$ on CIFAR-100. It is
lowest for label spaces in the middle of the range and highest at the two ends.
No plaintext model exists after HE-OFT ends.

The last two columns of
\cref{tab:headline} give the cost of the partition, $\Apool-\Adis$. On four of
the five tasks it exceeds the gap $\Adis-\Asel$ between the disclosed model and
HE-OFT. The cost of the partition is $0.182$ against
$0.071$ on AG-News, $0.257$ against $0.104$ on TREC, $0.163$ against $0.075$ on
Banking77, and $0.109$ against $0.029$ on CIFAR-100. DBpedia is the exception, at
$0.063$ against $0.136$. On those four
tasks, the partition of the data thus accounts for more of the distance between a
servable arrangement and a centralized model than never decrypting does.

\subsection{Handling Data Heterogeneity}
\label{sec:exp-select}

\begin{figure}[t]
\centering
\includegraphics{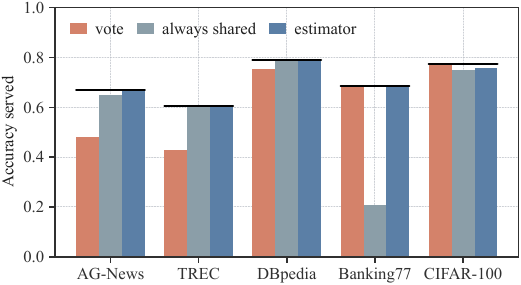}
\caption{Accuracy of the arrangement each selection rule serves, per task of
\cref{tab:headline}, averaged over three seeds. The black mark is the better of
the two servable arrangements, chosen per seed.}
\label{fig:regret}
\end{figure}

Under label skew, each client sees a
different part of the label space, and a model fitted to one client's data does
not transfer to the others. HE-OFT addresses label skew at three points,
and none of them needs a decryption. First, every client fine-tunes from the same
public pretrained backbone and the same public initializer. The backbone fixes
the features every client and every querier computes, which lets one head serve
all of them. Pretraining keeps the clients' solutions close enough for their
displacements to combine in one
round~\cite{ilharco2023editing,malinovsky2024oneround}. A common start alone does
not suffice under strong skew. In an earlier configuration that averaged the displacements of
heads started at random, Banking77 reached $0.378$ at $\alpha=0.05$, $0.443$ at
$0.1$, $0.585$ at $0.3$ and $0.669$ at $1.0$, against $0.887$ for pooled
training. Heads started from an initializer aligned on raw client examples held
$0.888$ to $0.891$ at every skew, but that initializer exposes the examples. Under
strong skew each client holds a few of the seventy-seven classes, and training
drives down the logits of the classes the client never sees. The average of ten such
displacements suppresses every class some client lacks, and one round does not
recover it. Second, the clients that hold a class
decide its row of the shared head. Every class
row has at least one holder in every cell we measured. Third, the federation
selects between the two arrangements under encryption (\cref{fig:regret}).

The simplest rule for choosing between the two arrangements fails
systematically. The estimator of \cref{sec:selection} corrects it and discloses
only the index of the selected arrangement.
Each client measures both arrangements on data held back
from its own training and the federation takes the sample-weighted majority. Across
all fifteen cells of \cref{tab:headline}, this procedure selects the personal
adapter every time, including on AG-News and TREC where the personal adapter is
$0.17$ and $0.18$ worse. Weighting each class equally within a client's held-out
set selects identically.

The two measurements estimate different quantities. The shared
arrangement is one model, and the union of the clients' data is the whole
training set. Pooled held-out measurements thus estimate its global accuracy to
within $0.02$ on average. The personal arrangement is $\Nc$ different models. Client $j$'s
held-out data measures only client $j$'s own model, on client $j$'s own
distribution, which is the distribution that model was fitted to. Its measured
accuracy exceeds its true global accuracy by $0.28$ on average. A client's
held-out set also contains only the classes that client holds. The client
therefore cannot measure locally the quantity that decides the comparison, its
accuracy on classes it has never seen. The rule compares an accurate number with
an optimistic one, and more held-out data does not remove the bias.

\Cref{tab:select} reports the estimator of \cref{eq:estimator} against that
baseline. Scoring both arrangements under the federation's own class prior selects
correctly in $12$ of $15$ cells. The shared arrangement pools evidence across
clients, and the personal arrangement's unobserved classes score zero.
Filling the unobserved classes from each client's rarest classes instead of with
zero selects correctly in $13$. Neither variant has a fitted threshold. Each compares two
numbers on the same absolute accuracy scale.

\begin{table}[t]
\centering
\caption{Choosing between the two servable arrangements without decrypting either.
Cells are counted over the five tasks and the four CIFAR-10 partitions, at three
seeds each. Regret is the mean accuracy a wrong choice gives up, over the cells
in which the rule chooses wrongly.}
\label{tab:select}
\footnotesize
\begin{tabular}{lcc}
\toprule
Selection rule & correct cells & regret \\
\midrule
majority of held-out accuracies      & 8/27 & 0.074 \\
same, class-balanced within a client & 11/27 & 0.085 \\
always the shared head               & 19/27 & 0.198 \\
global-prior estimator, zero fill    & 23/27 & 0.021 \\
global-prior estimator, rare-class fill & $\mathbf{24/27}$ & $\mathbf{0.020}$ \\
\bottomrule
\end{tabular}
\end{table}

Two further results concern the estimator. On Banking77, where a wrong
choice costs $0.48$, it selects correctly on all three seeds, by margins of $0.13$
to $0.24$. On the AG-News seed where the partition drops three clients, the shared
head reaches only $0.402$ and the personal adapter is better. The estimator
switches to the personal adapter there, which a rule keyed to the number of
classes would not do.

The estimator's estimate of the
shared arrangement's global accuracy is within $0.019$ on average across the
fifteen cells of the five tasks and within $0.005$ across the twelve CIFAR-10
cells, computed from encrypted local evidence without a test set.

The twelve CIFAR-10 cells separate the two rules more sharply than the five
tasks do. The shared head over the bare backbone is better in eleven of them. The majority
of held-out accuracies chooses the personal adapter in all twelve and is thus
correct once. The estimator is correct eleven times, and its one error costs
$0.006$. The vote does not fail on narrow margins. Each client measures the personal
arrangement on its own data, where that arrangement does best, and the vote adds
up such measurements.

\paragraph{Scope of the selected index} Each client could
instead adopt its own winner, which reveals one index per client and not one for
the federation. We measured that rule. It selects the personal arrangement almost
everywhere: ten clients out of ten on twelve of the fifteen cells, nine of ten on
two, and seven of seven on the AG-News seed that drops three clients. It reaches
$0.542$ on AG-News at seed $42$ where the shared arrangement reaches $0.809$. \Cref{sec:selection} names this bias. Per-client scoring keeps it, because a
client still measures the personal arrangement on its own distribution. The
winning arrangement does vary between clients, but a rule keyed to that variation
inherits the same optimism.

\paragraph{Packing and the setup traffic} A ciphertext at ring degree
$2^{15}$ carries $2^{14}$ slots, and the head has $\Cc d$ entries at $d=768$.
Four classes fit in one ciphertext and a hundred need five, which the server adds
independently. The cost of \cref{sec:exp-cost} thus grows with $\Cc$ through the
argmax and not through the additions. The head application is a slot-wise product
followed by a rotate-and-sum over $\log_2 d$ rotations. The rotation key set is
hence the ten power-of-two rotations that fold a $768$-entry block. \Cref{tab:comm} itemizes the $344$\,MiB each client sends at setup.
The relinearization and Galois shares dominate, and neither recurs.

\paragraph{Cost of selection} We measured the count of $2\Nc\Cc$
comparisons. At five parties and a hundred classes the circuit performs $1000$
scalar comparisons, which equals $2\Nc\Cc$, in $81$ sequential sign circuits, and
the whole phase takes $17.5$ minutes. The circuit does not compute one encrypted
argmax per held-out example. Each client reserves one held-out example per class, and those examples pack
into a single score ciphertext. We measured the selection of the
true-class logit with a plaintext mask. The protocol uses an encrypted one-hot
label instead, which adds one ciphertext product per client and arrangement. We measured this phase under collective
refresh and not under the server-side bootstrapping HE-OFT specifies. Its traffic
is thus an upper bound, and its wall clock a lower one.

\paragraph{Contents of the selection upload} Selection costs each
client a second encrypted message. It carries per-class counts of correct predictions on data
the client held back, a measurement of two finished candidates. No gradient,
update or intermediate state passes through it. One-shot in this paper means that
no intermediate training artifact is exposed, and the selection message contains
none.

\paragraph{Disclosure of the estimator} The federation decrypts one value,
the index of the adopted arrangement. The per-client per-class accuracies stay
encrypted. Together with the count matrix, they would describe each client's label
distribution alongside its model quality.

\subsection{Comparison With Plaintext and Differentially Private \mbox{One-Shot} Federated Learning}
\label{sec:exp-peers}

DENSE~\cite{zhang2022dense}, GH-OFL~\cite{turazza2026ghofl} and FedAUXfdp~\cite{hoech2022fedauxfdp} are the one-shot methods closest to HE-OFT
(\cref{tab:related}). They train other architectures from other starting points,
and each hands every participant a model in plaintext. We thus cannot rank
their accuracy against ours directly. We compare accuracy on their own
partitions, the accuracy each method gives up for its protection, and what its
released model gives an adversary. The peer group
evaluates on different tasks from ours, and we therefore ran HE-OFT on theirs. We adopt two
published partitions of CIFAR-10: $\Nc=5$ with $\alpha\in\{0.1,0.3\}$, the
configuration of~\cite{zhang2022dense}, and $\Nc=20$ with
$\alpha\in\{0.04,0.16\}$, the configuration of~\cite{hoech2022fedauxfdp}.

\begin{table}[t]
\centering
\caption{CIFAR-10 on two published partitions, three seeds, on the whole
$10{,}000$-image test set. The
disclosed model aggregates and decrypts both the adapter and the head. The last column is the share of the disclosed model's accuracy that HE-OFT
keeps without disclosing a model. All are measured on HE-OFT at each prior paper's partition.}
\label{tab:peers}
\footnotesize
\setlength{\tabcolsep}{5pt}
\begin{tabular}{cccccc}
\toprule
$\Nc$ & $\alpha$ & $\Aloc$ & $\Asel$ & $\Adis$ & $\Asel/\Adis$ \\
\midrule
$5$  & $0.10$ & $0.384$ & $0.949$ & $0.963$ & $0.986$ \\
$5$  & $0.30$ & $0.570$ & $0.952$ & $0.966$ & $0.985$ \\
$20$ & $0.04$ & $0.212$ & $0.950$ & $0.947$ & $1.003$ \\
$20$ & $0.16$ & $0.397$ & $0.953$ & $0.961$ & $0.992$ \\
\bottomrule
\end{tabular}
\end{table}

At $\Nc=5$ and $\alpha=0.1$ we measure $\Asel=0.949$, against the $0.503$
of~\cite{zhang2022dense}. Those methods train a ResNet-18 from scratch, and
every client here fine-tunes a frozen public ViT-B/16. Absolute accuracy across
papers therefore reflects the backbone. The gain over a client alone uses one
backbone on both sides, and at that partition it runs from $\Aloc=0.384$ to
$\Asel=0.949$.

The accuracy each method loses by protecting the contribution can be compared across papers. Every paper measures it against its own unprotected baseline on its own
architecture, and the backbone cancels. By this measure the peer group falls into
three cases. A method leaves the contribution unprotected, protects it with
differential privacy, or protects it cryptographically.

DENSE and FedDF send the contribution in plaintext and report no such cost. Each client sends the server a trained model in one case and
repeated predictions over many rounds in the other, both in plaintext.
\Cref{sec:related-leakage} shows that an adversary targets these messages.

FedAUXfdp protects its contribution with differential privacy. Against its own
undefended baseline on CIFAR-10, it loses $0.006$ of
accuracy at $\varepsilon=1.0$ and $0.626$ at
$\varepsilon=0.01$~\cite{hoech2022fedauxfdp}. Differential privacy bounds membership advantage by
$e^{\varepsilon}-1$~\cite{yeom2018privacy}, which constrains an attack once
$\varepsilon<\ln 2$, at budgets where FedAUXfdp trades more accuracy for privacy.
The bound is $1.72$ at $\varepsilon=1$, $0.65$ at $0.5$, $0.105$ at
$0.1$ and $0.010$ at $0.01$, and FedAUXfdp loses $0.006$, $0.029$, $0.413$ and
$0.626$ of accuracy at those budgets. It holds for the unit that FedAUXfdp's
mechanism protects, and the released model remains in every participant's hands
at any budget.

HE-OFT protects its contribution cryptographically. Against the same
computation in plaintext, encryption does not degrade accuracy, because the decrypted result
matches the plaintext computation to the encoding precision of the scheme. No method in the peer group withholds the model. HE-OFT keeps $98.5$ to $100.3$ per cent of the disclosed model's
accuracy over \cref{tab:peers}, and $85.3$ to $96.3$ per cent over
\cref{tab:headline}.

Every method in this group releases a model, and a participant holding it has a
copy at fidelity $1$ without a query. We measure membership on the model HE-OFT
itself would disclose, the counterpart of what these methods release. The
likelihood-ratio attack is strongest on CIFAR-100 and Banking77, where it
identifies four to nine times as many members as a random guess at a
false-positive rate of $1$ per cent. Every surface HE-OFT exposes stays close to a
random guess (\cref{tab:surfaces}). The attack reaches an area under the
curve of $0.690$ on CIFAR-100 and $0.579$ on Banking77, and at that
false-positive rate it identifies $7.2$ to $8.5$ and $4.3$ to $5.3$ per cent of
the members. On the three text tasks where we also attack HE-OFT, every surface
HE-OFT exposes stays between $0.47$ and $0.53$. GH-OFL sends per-class feature
statistics in plaintext, and these include the class totals that HE-OFT keeps
encrypted (\cref{sec:split}).

\paragraph{A partition where withholding costs nothing} At $\Nc=20$ and
$\alpha=0.04$ we measure $\Asel>\Adis$, by $0.003$ on average, on two seeds of
three. At that concentration each client holds close to a single class.
Averaging the adapters then pulls the representation toward every local task at
once and improves none of them. The arrangement that leaves the public backbone
unchanged avoids this, and the estimator selects it. In this cell the Dirichlet
draw leaves some clients with fewer than twenty samples, too few to train, and the
effective federation has $17$ to $19$ clients.

\paragraph{Selection on these partitions} The estimator of \cref{sec:exp-select}
selects correctly in $11$ of the $12$ cells behind \cref{tab:peers}. The exception
is $\Nc=5$ with $\alpha=0.3$ on one seed, where the selected arrangement reaches
$0.945$ and the other reaches $0.953$.

The disclosed model does not satisfy the requirements of
\cref{sec:setting}, and we do not present it as a configuration of this protocol.
Disclosure changes the threat model. A participant that holds the model in
plaintext can inspect it, run inference on it without limit, redistribute it, and
mount the white-box attacks that \cref{sec:related-leakage} surveys. The query
allowance of \cref{sec:exp-leak} then bounds nothing. A deployment that may
disclose the model should use a protocol designed for disclosure, and prior
one-shot federated learning offers several.

\subsection{Cryptographic Cost and Scalability}
\label{sec:exp-cost}

HE-OFT's traffic consists of a one-time setup and $13.5$\,MiB per query at $\Nc=10$. We implement HE-OFT in multiparty
CKKS on Lattigo~\cite{lattigo,mouchet2021multiparty} and simulate no
cryptographic step. All figures use ring degree $2^{14}$ for the shallow operations
and a deeper chain at ring degree $2^{15}$ for the circuits that need it. Timings
are single-run wall clock on one core of a VALAR CPU node. Communication figures are exact byte counts.

The parameter sets are chosen inside a $128$-bit budget. At ring degree
$2^{15}$ the residual chain is one $55$-bit prime over eight $45$-bit levels, for
$\log QP=537$, and the bootstrapping ring is $2^{16}$ at $\log QP=1480$. The
search enumerates candidate chains and rejects any whose bootstrapping modulus
exceeds the budget, and every latency below is measured at those parameters.

\paragraph{Cost of query encryption} Applying the head to an encrypted
query costs $39.0$\,ms against $5.4$\,ms for a plaintext one, a factor of seven.
Encrypting the query hides the feature vector from the server, which could otherwise invert it to recover the client's input. The factor is small beside the argmax that follows.

\paragraph{End-to-end query cost} \Cref{tab:costcmp} gives HE-OFT's best test time, at four classes, measured end to end from the query's encryption to the querier's decryption. \Cref{fig:trends}(d) gives the test time at every label-space size we measured, and it grows with the number of classes. The GPU figures use Phantom~\cite{yang2024phantom}, with level restoration measured on an NVIDIA H100. The key switch grows linearly with the number of clients but stays below one second at $\Nc=20$. A query that returns scores spends most of its time applying the encrypted head. A query that returns the label also computes the encrypted argmax, and restoring levels takes almost all of the argmax. The server restores the levels
alone, under bootstrapping keys that the clients generate together once, and no
restoration therefore needs a round with the clients. The
tournament needs $\lceil\log_2\Cc\rceil$ rounds against the $\Cc-1$ of the
sequential fold it replaces.

slytHErin, the closest published system for encrypted inference, evaluates a twenty-layer network under the same multiparty setting in $245.58$\,s at three parties and $354.17$\,s at twenty~\cite{intoci2023slytherin}. It runs batched on twelve cores and returns the score vector. Our $443.1$\,s and $1713.1$\,s cover one unbatched, label-only query on one core. A copy from labels costs an adversary $16$ to $260$ times as many queries as one from scores (\cref{sec:exp-leak}). Their batch
amortizes to about $0.95$\,s per sample at ten parties. Their cost lies in the
network and ours in the argmax, and the two figures thus describe the setting
without ranking the systems. They return a score vector and stop before the
reduction that makes up $96$ to $99$ per cent of our query. A system returning a
label would also pay for that reduction.

Three measures could lower that figure, and each acts on a different cost.
Multi-core execution does not shorten the circuit, because the rounds are
sequential and the comparisons within one round already form a single operation on
a single ciphertext. Multi-core execution applies inside each primitive, across the residue moduli
of the residue number system decomposition, and our single-threaded measurement
uses none of it. Overlapping computation with communication raises throughput but
not latency, because the level restorations lie on the critical path of one
query. A server that answers many queries can pipeline them. Only hardware
acceleration acts on the dominant term. The argmax is bootstrap-bound, and a GPU
implementation helps it most, although for the reason given below we do not
project a single number.

\paragraph{The encrypted reciprocal} Inverting the per-class totals under
encryption prevents any coalition from reading the remaining client's class
histogram (\cref{sec:split}). At $\Nc=10$ the inversion takes $4.21$\,s and needs two
collective refreshes. It agrees with the plaintext reciprocal to a relative error
of $1.9\times10^{-8}$, and it therefore does not degrade accuracy. The figure grows slowly with the
federation, $3.16$\,s at $\Nc=5$ and $6.30$\,s at $\Nc=20$, because only the
refreshes involve the clients. It is the only deep operation in building the shared head,
and it runs once per aggregation and not per query. A few seconds are negligible
beside the local training that precedes it.

\paragraph{Cost of the encrypted argmax} The argmax is the one deep operation of the serving
path. Evaluated as a sequential fold of $\Cc-1$ comparisons it takes about
$16$\,s per comparison, which is roughly twenty-six minutes at a hundred classes.
Packing the logits into one ciphertext and reducing them by a tournament lowers the
number of sequential comparisons from $\Cc-1$ to $\lceil\log_2\Cc\rceil$. The
tournament brings the same computation to under two minutes at a hundred classes, a
fourteenfold reduction. Both figures restore levels by collective refresh, and the
ratio thus isolates the effect of the tournament. Homomorphic evaluation under a
collectively generated key works as under a single key, and single-key results
therefore carry over. GPU implementations of the same reduction report a
full-vocabulary argmax in well under a second~\cite{zhang2025nexus}.

The tournament returns the largest logit, and HE-OFT needs its index. Carrying the index through the fold costs $4.4$ per cent at four classes
and $9.4$ per cent at a hundred, and the fold returns the correct index at every
class count we measured. We measured this twice, once under collective refresh and
once under the server-side bootstrapping HE-OFT specifies. The premium agrees to
within half a point, although the absolute latency differs thirteenfold. The cost
of the index thus depends on the fold and not on the mechanism that restores its
levels.

\paragraph{Ring degree and client count} The ring degree is the security
parameter, and the number of clients decides how many shares a decryption needs.
\Cref{fig:cost} separates them, from one measurement over the full grid.

The per-query arithmetic depends on the ring degree alone. It roughly
doubles with each doubling of the degree and stays flat in $\Nc$. The encrypted
product takes $38.4$, $39.0$ and $39.2$\,ms at $\Nc=5$, $10$ and $20$.

The key switch that returns a label depends on both, and doubles with each. It
runs from $103$\,ms at the smallest setting to $1.73$\,s at the largest. Traffic
follows the same law, at $2.0$\,MiB per participating client, because each
contributes one independent share and the server only sums them. The
encrypted reciprocal grows more slowly, $3.16$\,s at $\Nc=5$ against $6.30$\,s at
$\Nc=20$, since only its two refreshes involve the clients.

\begin{figure*}[t]
\centering
\includegraphics{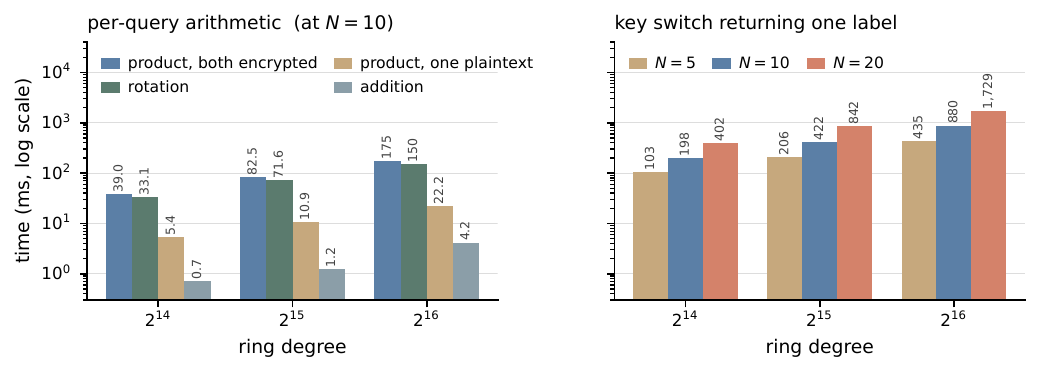}
\caption{Per-operation cost against the two axes that set it, from one
measurement over the whole grid. Left: the per-query arithmetic at $\Nc=10$,
which depends on the ring degree and not on the federation size. Right: the key
switch that returns one label, which grows with both. Read in seconds, the left
panel spans $0.0007$ to $0.18$\,s and the right panel spans $0.10$ to $1.7$\,s.
Neither is the cost of a query. Measured end to end, one query totals $443.1$\,s at four classes and $1713.1$\,s at a hundred, of which the encrypted argmax between them is $96$ to $99$ per cent.}
\label{fig:cost}
\end{figure*}

\paragraph{Hardware acceleration} The figures above are single-threaded CPU
figures. The argmax dominates the latency of a query. We therefore estimate what specialized hardware does to it. Since the cost of a CKKS operation grows with the
ring degree, the comparison is only meaningful at a fixed one. \Cref{fig:cost}
gives our figures, $39.0$, $82.5$ and $174.9$\,ms for a ciphertext-by-ciphertext
product at ring degrees $2^{14}$, $2^{15}$ and $2^{16}$, and $33.1$, $71.6$ and
$149.7$\,ms for a rotation. Phantom, a CUDA implementation of
CKKS~\cite{yang2024phantom}, reports the same primitives at ring degree $2^{16}$
on an A100. The product takes $1.61$ to $8.56$\,ms and the rotation $1.53$ to
$8.49$\,ms, over modulus chains of $29$ to
$41$ levels. We give no single ratio against those figures, because ours
are single-threaded Lattigo at a chain we did not match to theirs.

Level restoration is $94$ to $97$ per cent of a query and costs $22.6$ to $23.4$\,s per restoration on one core at every label-space size. A full-slot bootstrap costs $28.6$\,s in
OpenFHE at the same ring degree and $29$ levels~\cite{agullo2025fideslib}. We measure one restoration under Phantom on an NVIDIA
H100 at $2.28$\,s, at the ring degree $2^{16}$ HE-OFT restores at, a factor
of ten. Phantom publishes no bootstrapping benchmark, and no GPU library we found
restores levels under collectively generated keys. This figure therefore has no
published counterpart.

The accelerated figure is computed from that one measurement, not from a
ported implementation. A query performs $19$ restorations at four classes, $39$ at fourteen and $69$ at a hundred. The local evaluation that surrounds them, the head application included, costs $12.8$\,s, $29.1$\,s and $97.8$\,s on one core. Replacing each restoration with the H100 figure and leaving the local evaluation on the processor gives $19\times2.28+12.8=56.1$\,s, $39\times2.28+29.1=118.0$\,s and $69\times2.28+97.8=255.1$\,s. The measured end-to-end times on one core are $443.1$\,s, $919.3$\,s and $1713.1$\,s. We ported only the restoration.
The local evaluation is then the larger part of a query above four classes, and a
full port would have to move it as well.

\paragraph{One end-to-end query on a real head} Every figure above uses synthetic vectors. We therefore also checked that the encrypted path answers real queries as the plaintext model does. On the recorded AG-News
head at seed $42$, ten parties, ring degree $2^{15}$ and scale $2^{45}$, thirty-two
encrypted answers over the two servable arrangements all equal the plaintext
answers. The hardest case has a top-1 to top-2 gap of $0.037$, which the scaling
takes to $0.0068$, and its decoded index is within $3\times10^{-8}$ of the
integer. Every answer was correct, and the resolution limit thus lies below that gap.

Two of the sixteen answers under the shared arrangement are queries the
head itself gets wrong, and the encrypted path reproduces those errors.

\paragraph{Comparison with published work} We found no published
system that takes an argmax under encryption in a multiparty scheme, and there is
thus no matched number to compare with. The nearest match on the primitive is NEXUS, which
computes the same sign-based argmax under CKKS at the same ring degree
$2^{16}$~\cite{zhang2025nexus}. NEXUS returns an encrypted one-hot vector in $54$\,s over a label space of
$30{,}522$. Three conditions separate that figure from
ours. It is single-key, and it thus runs no collective key switching and no
smudging.
It uses $32$ threads where we use one core. It is also amortized over a batch of
$32$ inputs, where we serve one query.

POSEIDON answers
a query in $0.38$\,s at ten parties, at ring degree $2^{13}$ with six levels, and
returns the prediction vector after key-switching it to the
querier~\cite{sav2021poseidon}. slytHErin evaluates a twenty-layer network at
three model-providers with a latency of $245.58$\,s for a batch, at a throughput of
$1.19$ samples per second~\cite{intoci2023slytherin}. Two
ring-degree doublings and the argmax separate the first from our figure.
The second evaluates a whole network where we evaluate one linear map and
then compare its outputs. Our additional cost comes from computing the argmax and
returning only the label.

\paragraph{Against prior systems} \Cref{tab:costcmp} sets these figures beside the
costs the closest systems report. POSEIDON trains its networks under encryption,
and its cost grows with the model size~\cite{sav2021poseidon}. It keeps
every intermediate model encrypted and evaluates a deep circuit in every round.
SHE-LoRA fine-tunes a low-rank adapter over many rounds and encrypts a small part
of each update that the client selects~\cite{li2025shelora}. On
Llama-3.1-70B each client encrypts $0.125$ per cent of the parameters, at
$242.32$\,s of encryption and $385.84$\,MB of ciphertext per update. In its main experiments fifty
clients fine-tune OpenLLaMA-3B and BERT-Large over $200$ rounds. It sends the rest of
each update in plaintext and gives every client each round's aggregate under one
shared key, which keeps its encryption cost low. HE-OFT encrypts each
contribution once and forms no intermediate model. It trains in plaintext on each client and performs only
the aggregation under encryption. The aggregation takes a few seconds (\cref{tab:costcmp}) after one upload of $12$
to $44$\,MiB per client, and neither depends on the depth of the backbone. Both grow with the feature dimension $d$. At a hundred classes, raising
$d$ from $768$ to $4096$ takes the upload from $44$ to $212$\,MiB per client. The
server's sum of the uploads and its product with the reciprocal rise from $0.45$
to $2.38$\,s. The cost of a query does not grow with $d$, because the argmax that dominates it
takes only the $\Cc$ logits as input. Over the same range the cost changes by at
most $4.3$ per cent at each of four label-space sizes. The systems of \cref{tab:costcmp} return the score vector to the querier, which their threat models permit. HE-OFT also protects the served model from its querier, because scores let a querier recover a linear head from $d+1$ answers. It therefore returns only the label, which requires the argmax under encryption and accounts for its longer test time. A query that returns scores, as these systems do, costs HE-OFT about as much as one in CryptPEFT. The additional time is the cost of this stronger output guarantee, which \cref{sec:exp-leak} measures.

\begin{table}[t]
\centering
\caption{Cost against the closest systems, as published. Train time is the cost of producing the served model under encryption, and a dash means no training. Test time is the latency of one query, at HE-OFT's best case of four classes, with GPU figures projected (\cref{sec:exp-cost}). Output is what
the querier receives. From scores, $d+1=769$
answers recover HE-OFT's head, and from labels a copy of fidelity $0.90$ needs
$16$ to $260$ times as many (\cref{sec:exp-leak}). Hardware differs across rows. $^{*}$Extrapolated. $^{\dagger}$Per batch.}
\label{tab:costcmp}
\footnotesize
\setlength{\tabcolsep}{2.5pt}
\begin{tabular}{@{}llccc@{}}
\toprule
System & Model & \shortstack{Train\\time} & \shortstack{Test\\time} & Output \\
\midrule
POSEIDON ($\Nc{=}10$)~\cite{sav2021poseidon} & 3-layer MLP & $5{,}283$\,s & $0.38$\,s & scores \\
POSEIDON ($\Nc{=}50$)~\cite{sav2021poseidon} & CNN & $175$\,h$^{*}$ & $3.6$\,s$^{*}$ & scores \\
slytHErin & NN20 & -- & $245.6$\,s$^{\dagger}$ & scores \\
CryptPEFT & ViT-B/16 & -- & $2.26$\,s & scores \\
\midrule
\shortstack[l]{HE-OFT ($\Nc{=}10$),\\scores only} & RoBERTa-base & $4.21$\,s & \shortstack{CPU $2.6$\,s\\GPU $0.7$\,s} & scores \\
HE-OFT ($\Nc{=}10$) & RoBERTa-base & $4.21$\,s & \shortstack{CPU $443.1$\,s\\GPU $56.1$\,s} & label \\
\bottomrule
\end{tabular}
\end{table}

\paragraph{Communication} HE-OFT is one-shot in the sense of \cref{sec:intro},
since no intermediate training artifact is exposed, but the parties still
communicate. Key generation precedes training, selection costs at most
$2\Nc\Cc$ encrypted comparisons once, and serving costs traffic per query.

We measured the time and traffic of the selection, which the comparison
count does not give. At five clients and four classes it takes $400$\,s and
moves $3.4$\,GiB. At twenty clients and a hundred classes it takes $7111$\,s and
moves $145.4$\,GiB. The federation pays this once, before it answers anything,
and no part of it recurs per query.

The count stays within this bound, and the cost lies in the circuit depth. Our
implementation evaluates
$321$ sequential sign circuits at twenty clients and a hundred classes, where
$2\Nc\Cc$ allows $4000$ comparisons. Each
comparison under CKKS is a polynomial approximation of the sign. Every such
approximation consumes levels, and restoring them is $99.8$ per cent of the
traffic above. The current comparison literature has the same
cost~\cite{cheon2019comparison,lee2021minimax}, and we accept it in order not to
decrypt the two scores.

The selection is correct at every point we measured. It decodes the
winning arrangement in all twelve cells of the grid, and the per-class totals it
combines come back with a relative error near $10^{-8}$. \Cref{sec:exp-select}
reports which arrangement the rule chooses and how often it is right. 

\begin{table}[t]
\centering
\caption{Communication at $\Nc=10$, on the fifteen-modulus chain at ring degree
$2^{15}$ that the encrypted argmax requires. Share sizes do not depend on the
number of clients. The per-query total does, because the quorum returns one
key-switching share each. It does not depend on the number of classes,
because the level restorations are local to the server.}
\label{tab:comm}
\footnotesize
\begin{tabular}{llc}
\toprule
When & Item & Per client \\
\midrule
setup, once & collective public key share & $4.25$\,MiB \\
            & relinearization key share & $102.0$\,MiB \\
            & rotation keys, seven of them & $238.0$\,MiB \\
\midrule
training, once & encrypted head displacements & $12$ to $44$\,MiB \\
\midrule
per query & encrypted features, uploaded & $7.5$\,MiB \\
          & encrypted label, returned & $1.0$\,MiB \\
          & key-switching shares, ten & $5.0$\,MiB \\
\cmidrule(l){2-3}
          & \textbf{total per query} & $\mathbf{13.5}$\,\textbf{MiB} \\
\bottomrule
\end{tabular}

\end{table}

The training upload is the one row that moves with the number of classes.
A head displacement is $\Cc$ rows of $768$ values and a bias, and a ciphertext at
ring degree $2^{15}$ carries $16{,}384$ slots. A client therefore sends
$\lceil 769\Cc/16{,}384\rceil$ ciphertexts for each of the two arrangements and
one further ciphertext for its per-class counts. On the eight-modulus chain the
aggregation needs, a ciphertext is $4.0$\,MiB. That is three ciphertexts and
$12$\,MiB for AG-News, TREC, CIFAR-10 and DBpedia, nine and $36$\,MiB for
Banking77, and eleven and $44$\,MiB for CIFAR-100. It is paid once and it is
smaller than the rotation keys of the setup in every case.

The per-query total is $8.5 + 0.5\Nc$\,MiB, since the quorum returns one
key-switching share each. It is $11.0$\,MiB at $\Nc=5$ and $18.5$\,MiB
at $\Nc=20$. It is $6.8$\,MiB at ring degree $2^{14}$ and $27.0$\,MiB at
$2^{16}$, and it does not vary with the number of classes. Setup costs $344$\,MiB per client, once.

\paragraph{Growth with the federation size} Every per-party quantity in
\cref{tab:comm} is independent of the client count. We measured the grid at
$\Nc\in\{5,10,20\}$, at ring degrees $2^{14}$, $2^{15}$ and $2^{16}$, and on both the aggregation and the serving chain. The collective public key share, the
relinearization key share, the rotation key share, the ciphertext sizes and the
key-switching share are the same at every client count. The structure of the
scheme predicts this, since a party's share depends on the ring and the modulus
chain and not on the number of parties.

Two consequences set the cost of a large federation. What a client sends
does not grow as clients are added. Setup thus stays at $344$\,MiB per client,
and the federation's total key-generation traffic grows linearly in $\Nc$, which is
$3.4$\,GiB at $\Nc=10$ and $17$\,GiB at $\Nc=50$. A query costs the querier a fixed
$8.5$\,MiB and the quorum one key-switching share each, and per-query traffic
thus grows linearly in $\tc$ and not in $\Nc$. A deployment that lowers the
threshold for availability therefore also lowers the traffic on every query.
Its privacy cost is that a smaller coalition can decrypt, and \cref{sec:scope}
weighs this against availability.

The argmax consumes its
levels, and they must be restored several times per query. Restoring them by a
collective refresh requires a share from every client for every refresh, which at
a hundred classes amounts to roughly $1.6$\,GiB of traffic for a single label.
Generating bootstrapping keys once, collectively, and letting the server restore
levels on its own reduces that traffic to zero. The server can do so because evaluation under a
collectively generated key works as under a single key. The protocol specifies the
latter. We measured both, and each has a cost.

\Cref{tab:btp} reports the specified design beside the collective refresh, on which the earlier figures were measured. Restoring levels at the server makes the argmax about thirteen times slower at every label-space size, from $400.8$\,s against $31.2$\,s at four
classes to $1468.3$\,s against $113.0$\,s at a hundred. The bootstrapping keys
cost $8.70$\,GiB, generated collectively in $46$\,s, against a stored figure of
$15.5$\,MiB that measured the extended secret key rather than the key set the
server evaluates with. The saving is in traffic. One key generation of
$8.70$\,GiB removes about $1.6$\,GiB from every later query, and the two designs
thus break even near the sixth query, after which the specified one costs less.
Precision changes as well. The collective refresh reproduces the
plaintext maximum exactly, and the server-side bootstrap agrees to $7.3\times
10^{-5}$.

\begin{table}[t]
\centering
\caption{The two level-restoration mechanisms at $\Nc=10$, evaluation at ring
degree $2^{15}$. The specified design restores levels at the server under
collectively generated bootstrapping keys at ring degree $2^{16}$. The measured
alternative restores them by a collective refresh, which costs a share from every
client. Latency is of the encrypted argmax alone, single-run wall clock on one core of a VALAR CPU node.}
\label{tab:btp}
\footnotesize
\setlength{\tabcolsep}{4pt}
\begin{tabular}{lccccc}
\toprule
$\Cc$ & 4 & 6 & 14 & 77 & 100 \\
\midrule
collective refresh, s & $31.2$ & $47.3$ & $63.7$ & $112.4$ & $113.0$ \\
server bootstrap, s & $400.8$ & $612.5$ & $832.4$ & $1473.0$ & $1468.3$ \\
bootstraps & $17$ & $26$ & $35$ & $62$ & $62$ \\
error, server bootstrap & \multicolumn{5}{c}{$1.3$ to $7.3\times10^{-5}$, against exact for the refresh} \\
\bottomrule
\end{tabular}
\end{table}

\paragraph{Source of the reported latency} Every latency
above is the specified design, in which the server restores levels on its own
under the collectively generated keys. \Cref{tab:btp} sets it beside the
collective refresh, which is about thirteen times faster and costs a share from
every client at every restoration. We report the specified design because the method
describes it. A protocol that moved gigabytes to the clients on every
query could not serve the federation sizes this paper considers. Every figure is single-core, and the restoration is the part that
would gain most from threads.

\subsection{Extraction and Membership at Test Time}
\label{sec:exp-leak}

HE-OFT exposes no training-time artifact and no model. The only remaining
channel is the sequence of answers the served model returns, and the query
allowance bounds that channel but does not close it. No gradient, update or
data-derived summary is ever available in plaintext. There is one exchange, and
its single artifact is a ciphertext. No party holds the model in plaintext, and
the white-box attacks that a released model admits have nothing to attack. The adversary is thus a client that asks the served model questions and reads labels.

\Cref{tab:surfaces} compares what HE-OFT exposes with the surfaces that prior
systems expose, measured on the same heads over three tasks, three seeds and
both arrangements. A system that releases the model gives a copy at fidelity $1$
without a single query, and the likelihood-ratio membership attack against the
disclosed model is strongest on the two large label spaces, CIFAR-100 and
Banking77. A system that returns scores lets a client solve for a linear head from
$d+1$ answers. Multi-round FL exposes every client's
updates on every round~\cite{zhu2025fedmia}. HE-OFT exposes labels alone, and a
copy at fidelity $0.90$ then costs $16$ to $260$ times the queries it costs from
scores.

\begin{table}[t]
\centering
\caption{What each surface gives an adversary, on HE-OFT's heads over three
tasks, three seeds and both arrangements. The model-in-plaintext row adds
CIFAR-100. The per-round row gives its source's setting and true-positive
rate (TPR) at a false-positive rate (FPR) of $0.1$ per cent. Queries is the cost
of a copy at the fidelity shown, and a dash marks a column that does not apply.
Fidelity measures how much of the hidden model an adversary obtains. Membership,
the area under the curve of our strongest attack, measures what the surface
reveals about the training data, and $0.5$ is a random guess.}
\label{tab:surfaces}
\footnotesize
\setlength{\tabcolsep}{3pt}
\begin{tabular}{@{}lccc@{}}
\toprule
Surface & Queries & Fidelity & Membership \\
\midrule
\begin{tabular}[t]{@{}l@{}}model in plaintext\\\cite{zhang2022dense,turazza2026ghofl,hoech2022fedauxfdp}\end{tabular}
  & -- & $1$ & $0.50$ to $0.69$ \\
\begin{tabular}[t]{@{}l@{}}scores per query\\\cite{sav2021poseidon,intoci2023slytherin,xia2026cryptpeft}\end{tabular}
  & $769$ & $1.000$ & $0.49$ to $0.53$ \\
updates per round~\cite{zhu2025fedmia}
  & -- & -- & \begin{tabular}[t]{@{}c@{}}TPR $25.3\%$\\at FPR $0.1\%$\end{tabular} \\
\midrule
labels only, HE-OFT
  & \begin{tabular}[t]{@{}c@{}}$1.2\times10^4$ to\\$2.0\times10^5$\end{tabular} & $0.90$ & $0.47$ to $0.52$ \\
\bottomrule
\end{tabular}
\end{table}

\begin{figure}[t]
\centering
\includegraphics{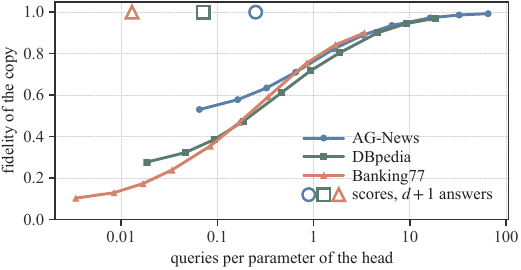}
\caption{Fidelity of a copy of the served head against the queries spent per
parameter of the head, mean over three seeds and both arrangements. Each line is
one task under the label-only interface that HE-OFT serves. Each hollow marker,
in the colour and shape of its task, is an interface that returns scores, where
solving from $d+1$ answers recovers the head exactly.}
\label{fig:extraction}
\end{figure}

Two different bounds apply to this channel. The first is cryptographic. By
\cref{thm:semihonest}, HE-OFT's messages are computationally independent of the
private data. The second is not cryptographic.
A client that queries the served model at points of its own choosing can, with
enough queries, reconstruct a functional copy of the shared head. Returning the
label instead of the logits removes the linear solve that would recover the head
in as many queries as the feature dimension. This raises the cost of a copy
without making it infinite. We therefore state the guarantee as a bound on
queries, and we measure what that bound has to be.

\paragraph{Cost of extraction} We give the adversary every advantage the
protocol admits. It is a participating client. It computes query features itself,
may submit an arbitrary vector in place of a real input, and is charged one query
per label. Submitting such a vector deviates from the protocol, and the adversary
is thus a malicious client whose deviation the encryption hides from the server.
The threat model does not fix what an honest client queries. A client whose
inputs repeat learns little from each further label, and a client whose inputs
vary widely comes closer to this attack. The adversary fits a linear model to the labeled queries, and we score the copy by \emph{fidelity}, the rate at which it agrees with the served
model on held-out features. \Cref{fig:extraction} reports the result. A constant
copy reaches the share of the most common label, $0.488$ on AG-News, $0.250$ on
DBpedia and $0.083$ on Banking77.

Extraction is possible, and we measure its cost. Reaching fidelity $0.90$ takes
about $1.2\times10^4$ queries on AG-News, $5.0\times10^4$ on DBpedia and
$2.0\times10^5$ on Banking77. Those figures follow the size of the head and not the
task. The head has $\Cc d$ parameters. In each case fidelity $0.90$ costs
between three and five times that number of queries, with $0.80$ costing roughly
$1.5$ times and $0.95$ roughly ten times. The multiplier thus depends on the
target fidelity. Writing $\varepsilon$ for the shortfall from a perfect
copy, our three tasks fit $Q\approx 0.4\,\Cc d/\varepsilon$ across the measured
range. For a linear head on features of this kind, a deployment can therefore set its
allowance from $\Cc$, $d$ and the fidelity it concedes, without
repeating this experiment. An
allowance set for one target does not carry to another. A published
measurement of the same interface is consistent with this. Extracting a softmax model from class
labels alone costs $100$ queries per parameter for agreement above $99.9$ per
cent~\cite{tramer2016stealing}, and the curve here reaches $0.993$ at $65$
queries per parameter on AG-News.

The allowance follows from these figures. A coalition of $\tc-1$ clients,
the largest that cannot decrypt, commands $(\tc-1)Q$ queries. Holding it below
fidelity $0.90$ at $\Nc=10$ thus requires $Q$ below roughly $1.3\times10^3$ per client on AG-News and
$2.2\times10^4$ on Banking77. The small label space is the binding case, because a head with few rows is cheap
to copy. The shared arrangement is also preferable there (\cref{tab:headline}),
and the deployments that gain most from the federation must therefore meter
queries most tightly.

Fidelity measures agreement with the served model. A copy at fidelity $0.90$ reproduces nine in ten of the
served model's decisions, including its mistakes, but the adversary has not
thereby learned any client's data. The recovered object is a linear map
over a public feature space, fitted from labels the adversary requested at points
it chose. No client record enters the attack at any stage. Turning a
copy into a statement about a particular record is the separate problem of
membership inference or reconstruction. That problem needs a surface HE-OFT
does not expose, since no gradient, update or data-derived summary is ever
available in plaintext. \Cref{sec:c2c} measures membership inference and
finds it close to a random guess.

\paragraph{The quorum's view of a query} The functionality tells the server that client $j$ asked a query, and HE-OFT tells the quorum the same. The label is key-switched to client $j$'s public key, and every contributing client sees whose key it is. A federation of competitors should
therefore treat the query pattern, who asked and how often, as disclosed to the
quorum as well as to the server. The query, the features and the label stay
encrypted to everyone but the querier. Hiding the pattern needs a further
mechanism, such as routing every answer through an anonymous channel, which costs
a round and is independent of the cryptography here.

\paragraph{Enforcement of the query allowance} The allowance $Q$ is a functionality
parameter, and a deployment has to place the counter somewhere. The server sees
every query and is the obvious place. \Cref{def:semihonest} lets it be corrupt,
however, and a corrupt server can refuse to count. This does not break the bound
of \cref{thm:malicious}, which already assumes an honest server, and
\cref{sec:malicious-ext} shows that this assumption is necessary. Under a
threshold below $\Nc$, a client can also spread its queries across quorums, and
each quorum sees only its own share. A per-quorum tally is then not a per-client
tally. A deployment that needs the bound enforced must
either keep $\tc=\Nc$ or run the counter
as a separate service that the clients trust for metering alone. Our measurements use $\tc=\Nc$.

\paragraph{A second query strategy} We also ran a
variant that spends half its budget bisecting toward the decision boundary and
then fits the same linear model. The variant is worse in $13$ of the $15$ cells, reaching
$0.675$ against $0.755$ on Banking77 at $5\times10^4$ queries. A
bisection chain returns many nearly identical points, and a linear fit gains less
from them than from independent ones. Tramer et al. report the same ordering for
multiclass models and offer the same reason, that the searches split across
several decision boundaries~\cite{tramer2016stealing}.

The ordering reverses once the budget is large relative to the head. On
AG-News the variant crosses over at about $16$ queries per parameter and at $65$
it leaves $0.0041$ of the decisions wrong against the random strategy's $0.0072$,
a factor of $1.77$. Banking77 never reaches $3.4$ queries per parameter at our
largest budget. Its cells thus lie below the crossover, and the count of $13$
reflects the budget grid and not the strategy. We report the random
strategy because it is the stronger one over the range in which an allowance would fall,
and because it is the simpler.

Neither strategy is the attack a cryptanalyst would mount, since both fit
a model to labeled points. The published constructions solve for each decision
hyperplane from points found on it, and they tie the hyperplanes to one another to
fix, once, the scale that the argmax
discards~\cite{lowd2005adversarial,carlini2025hardlabel}. That route costs a
factor logarithmic in the precision demanded, where ours is linear. Our figures
therefore bound the allowance from the wrong side, and they state what a
straightforward adversary pays.

\paragraph{Noise on the returned label} The query allowance is one way to
make extraction costly, and perturbing the answer is another. We measure the
second. The server flips the
returned label under randomized response at budget $\varepsilon$, which is the
label-level mechanism~\cite{mcsherry2007mechanism}. We record both the fidelity of the
copy and the accuracy the served model retains.
\Cref{tab:defence} reports the result.

\begin{table}[t]
\centering
\caption{Randomized response on the returned label. Copy fidelity after
$10^5$ label queries, against the accuracy the served model retains, at each
privacy budget $\varepsilon$ in the column heads, where $\varepsilon=\infty$
means no randomization. Mean over three seeds. The majority share a constant
predictor reaches is $0.488$ on AG-News, $0.250$ on DBpedia and $0.083$ on
Banking77.}
\label{tab:defence}
\footnotesize
\setlength{\tabcolsep}{3pt}
\begin{tabular}{llcccccc}
\toprule
Task & & $\infty$ & $8$ & $4$ & $2$ & $1$ & $0.5$ \\
\midrule
AG-News   & fidelity & $0.989$ & $0.986$ & $0.954$ & $0.889$ & $0.774$ & $0.598$ \\
          & accuracy & $0.649$ & $0.648$ & $0.620$ & $0.494$ & $0.372$ & $0.303$ \\
\midrule
DBpedia   & fidelity & $0.956$ & $0.946$ & $0.831$ & $0.571$ & $0.268$ & $0.139$ \\
          & accuracy & $0.789$ & $0.785$ & $0.640$ & $0.297$ & $0.150$ & $0.103$ \\
\midrule
Banking77 & fidelity & $0.777$ & $0.722$ & $0.364$ & $0.052$ & $0.020$ & $0.016$ \\
          & accuracy & $0.206$ & $0.203$ & $0.092$ & $0.028$ & $0.018$ & $0.015$ \\
\bottomrule
\end{tabular}
\end{table}

The mechanism lowers fidelity and accuracy together. At $\varepsilon=8$ it
changes neither, costing AG-News $0.001$ of accuracy and the copy $0.003$ of
fidelity, and below that both fall together. On AG-News the budget that first costs
the copy a fifth of its fidelity is $\varepsilon=1$. At that budget the copy still reaches
$0.774$ and the served model has fallen to $0.372$, under the $0.488$ share of its
most common class. On DBpedia accuracy falls by a larger fraction than fidelity at
every budget below $\varepsilon=8$. On Banking77 the copy degrades faster than the
task below $\varepsilon=4$, but only after accuracy has reached $0.028$ against a
majority share of $0.083$. The model the noise protects is then already worse
than a constant predictor.

This result favors metering queries over perturbing answers. Label noise
costs every honest query and every extraction query at the same rate, and an
extraction campaign spends $10^5$ queries where an honest client spends one. The
allowance charges the two differently, and the noise does not. We therefore state
the operational bound as a bound on queries and add no noise to the answers.

The defenses we survey do not apply under encryption. Query-distribution monitoring reads the query and tests whether a
client's stream looks like a
sweep~\cite{juuti2019prada}. Adaptive misinformation reads the query and answers
differently when it looks out of distribution~\cite{kariyappa2020adaptive}. Proof
of work reads the query to set its cost~\cite{dziedzic2022proof}. Watermarking needs
the answer to be a deterministic keyed function of the query, so that a repeat
cannot be filtered out~\cite{szyller2021dawn}. Our server holds a ciphertext of a
feature vector and, under the same assumption that protects the head, learns
nothing from it. Each of these defenses would need the server to decrypt the
query, which is the argument of \cref{prop:nogate} applied to a different gate.

A theorem also shows that output randomization is weak.
Against a server that answers incorrectly with constant probability $\rho$ below
one half, extraction costs a factor of $8/(1-2\rho)^2$ times a logarithm and no
more~\cite{chandrasekaran2020exploring}. The theorem bounds what the defense achieves, and
\cref{tab:defence} adds what the defender pays, which their analysis leaves out.
One mechanism satisfies both restrictions, and we have not measured it. Drawing a
fresh head from a distribution for each query needs no access to the query and
changes only answers near a boundary~\cite{alabdulmohsin2014adding}. It has a known break at a higher query
cost, and its accuracy cost under our arrangement is open.

\paragraph{Extraction against a released model} A released model is a
different case. A participant holding the model in plaintext has fidelity $1$ at
zero queries, can read the parameters directly, and faces no allowance, because
there is nothing to meter. Returning logits would be nearly as generous. The head
is linear in features the client computes itself, and $d+1$ logit queries recover
it by a linear solve, $769$ queries here against the $1.2\times10^4$ to
$2.0\times10^5$ that labels demand. We confirm this directly, recovering the head to fidelity $1.000$ from
$769$ answers when the served interface returns a distribution instead of an
index. Returning only the label turns extraction from a solve into a search, and the
query allowance sets the cost of that search. Deployed systems reached the same
conclusion. Recovering a production model's final projection layer became harder
once the providers withdrew the richer outputs that HE-OFT never
offers~\cite{carlini2024stealing}.

\section{Related Work}
\label{sec:related}

\begin{table}[t]
\centering
\caption{The closest systems on four axes, as each paper reports them. The model column names the largest model each paper evaluates. The plaintext column says whether any party ends up holding the
trained model in plaintext, and \emph{provider} means one party does while the
querier does not. A dash in the rounds column marks slytHErin and CryptPEFT,
which serve a model they do not train. HE, DP and 2PC are homomorphic
encryption, differential privacy and two-party computation. POSEIDON keeps the model under the
collective key and states that the querier does not obtain the model's weights.}
\label{tab:related}
\footnotesize
\setlength{\tabcolsep}{2.5pt}
\begin{tabular}{llcccc}
\toprule
System & Model & Rounds & Protection & \shortstack{Model in\\plaintext} & \shortstack{Querier\\receives} \\
\midrule
DENSE~\cite{zhang2022dense}          & ResNet-18   & one  & none & yes      & model \\
GH-OFL~\cite{turazza2026ghofl}       & VGG-16      & one  & none & yes      & model \\
FedAUXfdp~\cite{hoech2022fedauxfdp}  & MobileNetV2 & one  & DP   & yes      & model \\
SHE-LoRA~\cite{li2025shelora}        & Llama-3.1-70B & many & partial HE & yes & model \\
POSEIDON~\cite{sav2021poseidon}      & CNN         & many & HE   & no       & scores \\
slytHErin~\cite{intoci2023slytherin} & NN50        & --   & HE   & no       & scores \\
CryptPEFT~\cite{xia2026cryptpeft}    & ViT-B/16    & --   & 2PC  & provider & scores \\
\midrule
HE-OFT                               & RoBERTa     & one  & HE   & no       & label \\
\bottomrule
\end{tabular}
\end{table}

We compare the closest systems in \cref{tab:related}.
The multi-round cryptographic methods pay the encryption cost on every round and protect a per-round
sum. The one-shot methods send their single contribution in plaintext, or add
noise to it and lose accuracy. The federated fine-tuning methods send a small
object and protect none of it. One prior scheme is one-shot, federated and
encrypted, but only for a single-layer closed-form learner that cannot adapt a
deep representation~\cite{fedheonn2023}. No prior system meets three requirements
at once. A protocol that meets them exposes no training-time quantity, exchanges
only once for training, and never discloses the model, even to the parties that built
it.

\subsection{Attacks on Federated Learning}
\label{sec:related-leakage}

Federated learning (FL) began as an iterative procedure in which a server
repeatedly broadcasts a model, clients train it on their own data, and the server
averages the returned updates~\cite{mcmahan2017communication}. In every round each
client sends a new update computed on its private data. Several attacks show that
these updates leak more than the final model. An attacker can
reconstruct a batch from the gradient it
produced~\cite{zhu2019deep,geiping2020inverting}. An observer of the per-round
updates infers membership and unintended properties of another client's data, and
the attack is strong because the updates
repeat~\cite{nasr2019comprehensive,melis2019exploiting,zari2021efficient}. A server
that chooses the weights it broadcasts needs no inference, since crafted weights
make clients return verbatim copies of their own inputs, even through a
secure-aggregation layer~\cite{boenisch2023curious,fowl2022robbing}. Several works
share predictions instead of gradients, and the predictions also
leak~\cite{takahashi2023breaching,yang2023fdleaks,gu2023ldia,shi2025unveiling}.

Attacks that see only the final model are weaker, and Kairouz et al.
separate the two settings~\cite{kairouz2021advances}. Several works
infer membership from a released model, or invert it from released
confidences~\cite{shokri2017membership,yeom2018privacy,salem2019ml,fredrikson2015model,hitaj2017deep}. On a model that generalizes well, these attacks identify few training members
at a low false-positive rate~\cite{carlini2022membership}. An update is a detailed
view of one client's data and arrives every round. The final model is a single
average, and generalization moves it away from any one training example. A
protocol that never sends an update, and that exchanges only once, hence
leaves every adversary with the weaker attack.

\subsection{Protecting Training With Cryptography}
\label{sec:related-crypto}

Secure aggregation lets the server learn only the sum of the clients' updates,
using masks that cancel in the aggregate~\cite{bonawitz2017practical}. It is
lossless and protects one round's sum inside an iterative protocol. Kerkouche
et al. and So et al. show that the revealed per-round sums leak across
rounds~\cite{so2023securing,kerkouche2023client}. Under homomorphic encryption, the server computes
on contributions that stay encrypted. POSEIDON runs encrypted
stochastic gradient descent under multiparty CKKS, with activations replaced by
polynomial approximations so that the server can evaluate them
homomorphically~\cite{sav2021poseidon}. The model
loses no accuracy. The approach is iterative, and each encrypted
nonlinearity uses a deep multiplicative circuit with repeated bootstrapping. Several works study
those polynomial
activations~\cite{garimella2021sisyphus,baruch2022methodology,agamennone2025polynomial,alhossain2025training,ibarrondo2021fhebn}.
The cost grows with the network being trained. HE-OFT trains in plaintext on each
client and encrypts the aggregation, which lets it fine-tune a pretrained
transformer.

Several works encrypt the aggregation step of a plaintext training
loop~\cite{zhang2020batchcrypt,jin2023fedml,wei2025fedshe,li2025shelora,kemmaka2025hybagg}.
They run over many rounds, encrypt the contribution of every round, and protect
each round's sum. Single-key homomorphic encryption assumes that the server and
the clients do not collude~\cite{li2025shelora}, and HE-OFT uses the threshold
structure of \cref{sec:mhe} in place of that assumption.

Transfer learning has also been run under encryption. Several works train a
classifier head, a linear or logistic model, or a low-rank adapter on frozen
features under
CKKS~\cite{lee2023hetal,kim2018securelr,kim2018logistic,chiang2025cnn,privfedtl2026,alamin2025vit,li2024privtuner,frery2025private}.
Their setting, a frozen public backbone with a small trainable unit, is close to
ours. They run the optimizer on ciphertexts, with a multiplicative depth that grows
with the number of steps and a polynomial for every nonlinearity. HE-OFT runs the
optimizer in plaintext on each client and encrypts only its result.
Several of them also have a single data owner who outsources
computation. Our parties are mutually distrustful, train on their own data, and
must protect what they trained.

\subsection{One-Shot Federated Learning and Federated Fine-Tuning}
\label{sec:related-oneshot}

One-shot federated learning replaces the many rounds with a single
exchange~\cite{guha2019oneshot,wang2025towards}. Several works exchange predictions
on a shared public set, which in practice repeats the
exchange~\cite{li2019fedmd,itahara2021dsfl,lin2020feddf}. Other methods use a
single round. They distil the clients' models through a generator or synthesized
inputs, or fuse the models, in one round or in stages the authors count as
one~\cite{zhang2022dense,dai2024coboosting,liu2024fedlpa,zhang2024fedsd2c,tang2024fusefl}.
They show that one round can produce a usable model, with each client sending an entire model or a synthetic distillate of its data in plaintext. The attacks of \cref{sec:related-leakage} target such contributions.
In GH-OFL, each client runs a frozen pretrained encoder and uploads
per-class sufficient statistics once~\cite{turazza2026ghofl}. Those statistics
travel in plaintext, including the per-class totals, which stay encrypted in HE-OFT (\cref{sec:split}).

The one-shot methods that address privacy use differential privacy. FedAUXfdp
releases client contributions under a differentially private mechanism after
distillation through a pretrained extractor~\cite{hoech2022fedauxfdp}. Related
approaches protect a diffusion-generated surrogate or a teacher
ensemble~\cite{feddiff2024,li2021fedkt}, or apply a noise-free mechanism to the
distillation itself~\cite{sun2021fedmdnfdp}. These are our closest quantitative peers. They choose differential privacy and add the noise a privacy budget requires to the released artifact. HE-OFT chooses encryption, which costs the merge no accuracy. Differential privacy gives a
statistical bound on what a released artifact can reveal, at a cost in accuracy.
Encryption makes a contribution computationally indistinguishable from random to
every party without the key, and the merge loses no accuracy. Several works
apply differential privacy to transfer learning, running noisy gradient descent on
frozen features or on a low-rank
adapter~\cite{tramer2021features,mehta2022dpfeatures,yu2022dpft,li2022dplearners,liu2023dplora,xu2024dpdylora}.
They share our structure of a backbone and an adapter. Their released model trades accuracy for the privacy budget, and they train centrally or over many rounds.

Federated parameter-efficient fine-tuning sends a low-rank adapter in
place of the whole model. FedIT averages the adapters round by round~\cite{zhang2024fedit}, which approximates the average of their products. A client learns a product of two factors, and averaging the factors separately does not give
the average of the products. Several works correct that bias by stacking factors,
by reconstructing full products, by reconciling heterogeneous ranks, or by sharing
one factor
only~\cite{wang2024flora,bai2024flexlora,cho2024hetlora,guo2025fedsalora}.
FFA-LoRA freezes the randomly initialized down-projection and trains only the
up-projection. Averaging the trained factor is then exact~\cite{sun2024ffalora}.
Combining displacements from a shared initializer is task-vector merging, which
is equivalent to one-shot federated averaging. For foundation models in
plaintext, one round of fine-tuning
suffices~\cite{ilharco2023editing,tao2024taskarith,malinovsky2024oneround}. These methods send their adapters in plaintext, and HE-OFT encrypts its contribution.

\subsection{Secure Inference}
\label{sec:related-inference}

Secure inference answers queries from a model that stays encrypted. Several
works evaluate a transformer under encryption, either
non-interactively under CKKS or between two or three parties with secret sharing.
In each case the nonlinearities of attention dominate the
cost~\cite{zhang2025nexus,hao2022iron,pang2024bolt,dong2023puma}. The map HE-OFT serves is linear in the features. A trained map inside the
backbone would instead reduce serving to the problem these systems address.

Two systems are closer to HE-OFT. slytHErin evaluates a network under CKKS with the model
held under a multiparty key, and it key-switches the prediction to the querier~\cite{intoci2023slytherin}. The
same group built POSEIDON.
Together, the two systems come closest among prior work to a model that no party
decrypts. slytHErin evaluates the whole network homomorphically. The models it
reports are thus a twenty-layer and a fifty-layer network on a handwritten-digit
task, and not a pretrained encoder. It returns the prediction vector, and the
querier hence receives scores. By \cref{sec:exp-leak}, scores are a cheaper target
for extraction than a label.

CryptPEFT divides the work as HE-OFT does. A public frozen backbone runs in plaintext on the client, and the
client encrypts the resulting features. The provider evaluates only the adapter under
encryption, and the client receives the adapter's
output~\cite{xia2026cryptpeft}. HE-OFT likewise confines encrypted computation to
a small unit after the representation. CryptPEFT is a two-party protocol in which
the trained unit is the provider's plaintext, and it serves an adapter trained elsewhere. It
reports one inference in $2.26$\,s on CIFAR-100, and thus shows that serving a small encrypted unit after the backbone takes seconds. Because the adapter is one party's plaintext, CryptPEFT evaluates
no ciphertext-by-ciphertext product and restores no levels.

Two constructions for the encrypted maximum are faster than the tournament
of \cref{sec:serving}. One computes ranking and order statistics at constant
comparison depth, and the other decodes a language model without decrypting the
scores~\cite{mazzone2025ranking,avitan2025cutmax}. Both are single-key and omit
the collective refreshes that the threshold setting requires, and hence
neither transfers unchanged.

\section{Conclusion}
\label{sec:conclusion}

We presented HE-OFT, a federated fine-tuning protocol in which no party receives
the trained model. FL leaks mostly
through the quantities a protocol exposes during training, and some deployments
may not distribute the finished model. Each client
fine-tunes on a frozen public backbone, keeps its adapter, and uploads a single
encrypted head displacement. The server combines those displacements under
multiparty CKKS and never decrypts the result. A quorum of clients returns only the predicted label to the querier.

Because no party holds the model, the querier computes the features of its query
itself, with the public backbone and its own adapter. HE-OFT therefore shares
only a trained map after the last nonlinearity of the network. Across four text
classification tasks and one vision task, HE-OFT is more accurate than any client
alone. It keeps $85$ to $96$ per cent of the accuracy of a model that is
aggregated and disclosed.

\label{sec:scope}

\paragraph{Linearity of the shared quantity}
The requirement never to disclose the model forces this linearity
(\cref{sec:split}), and it is a limitation. A shared adapter inside the
backbone would adapt every participant's representation, where a local adapter
adapts only its owner's. Serving such an adapter reduces to secure transformer
inference (\cref{sec:related-inference}) and would carry its cost.

\paragraph{A public constant in the serving path} The serving circuit
compares logits through a polynomial approximation of the sign, and that
approximation resolves a fixed absolute margin. Multiplying every logit by a positive
constant preserves the argmax but changes the margin the circuit must resolve. We fix one such constant from the exported head. A deployment
holds the head only as a ciphertext and thus needs a public bound on its norm. We treat that bound as public (\cref{sec:security}) and do not
measure the cost of a loose bound.

\paragraph{Bounds on model extraction} The query
allowance of \cref{sec:exp-leak} is an operational control. Where no allowance separates honest use from extraction, calibrated noise
on the returned labels bounds the cumulative leakage
instead~\cite{mcsherry2007mechanism}. The mechanism composes with HE-OFT without
change. HE-OFT protects the contributions cryptographically, and differential
privacy would protect the answers statistically. None of our accuracy figures uses the mechanism.

\paragraph{Malicious participants} The cryptographic guarantees assume
semi-honest parties. A client that uploads a crafted head displacement can
bias the shared head. The server cannot inspect the upload because it is
encrypted. Two defenses are compatible with the encryption, and we leave both to
future work. One is a norm bound on each upload, enforced by a zero-knowledge
proof. The other is robust aggregation evaluated under encryption.

\paragraph{Availability} We state HE-OFT for a
$\tc$-out-of-$\Nc$ threshold (\cref{sec:security}) and take every measurement at $\tc=\Nc$.
This setting gives the strongest collusion resistance and the weakest
availability, since every client must take part to produce a label. Lowering
$\tc$ improves availability, but a coalition of $\tc$ clients can then decrypt. The query allowance need not tighten, because the bound of
\cref{sec:exp-leak} scales with $\tc-1$, the size of the largest coalition that
cannot decrypt. The
deployment chooses $\tc$, and we measured only $\tc=\Nc$.

The scheme admits any $\tc$ with $2\le\tc\le\Nc$. At a general threshold,
a query needs only $\tc$ of the $\Nc$ clients, and serving continues while a
participant is offline. Our implementation sets $\tc=\Nc$. In the configuration we
measured, a single unavailable client hence blocks every query until that client returns. A
deployment that needs serving to continue through churn sets $\tc<\Nc$ and accepts
coalitions of that size.

\Cref{sec:c2c} bounds the query channel at every threshold. Lowering
$\tc$ lowers the total allowance $(\tc-1)Q$ of the largest coalition that cannot
decrypt. The advantage $\delta_{\mathrm{wb}}$ does not change, since it is a
property of the shared head. The threshold therefore trades availability
against the size of the coalition that can decrypt.

We identify two promising directions for future work. The first is robustness to actively
malicious participants, which the present guarantee does not cover. Thus, robust
aggregation over encrypted contributions remains an open challenge. The second is a shared adapter inside the backbone,
which would require evaluating the transformer under encryption at every query.

\section*{Acknowledgments}
We acknowledge the support of T\"UB\.{I}TAK (the Scientific and Technological Research Council of T\"urkiye)
projects 124E091 and 124N941. The authors used Claude models to polish the manuscript.

\bibliographystyle{IEEEtran}
\bibliography{refs}

\end{document}